\documentclass[10pt,letterpaper]{article}

\usepackage{bm}
\usepackage{geometry}
\usepackage{graphics,epsfig,rotate,lscape,graphicx,amsmath,amsthm,amssymb,float,amsfonts,amsbsy,delarray,sectsty,amsfonts,amscd,pifont}
\usepackage{color,multirow}
\usepackage{algorithmic}
\usepackage{algorithm}

\allowdisplaybreaks

\newtheorem{theorem}{Theorem}
\newtheorem{lemma}{Lemma}

\newtheorem{claim}{Claim}

\newtheorem{definition}{Definition}
\newtheorem{proposition}{Proposition}

\newcommand{\bea}{\begin{eqnarray*}}
\newcommand{\eea}{\end{eqnarray*}}
\newcommand{\ed}{\end{document}}

\newcommand{\btab}{\begin{tabular}}
\newcommand{\etab}{\end{tabular}}
\newcommand{\bc}{\begin{center}}
\newcommand{\ec}{\end{center}}

\newcommand{\bi}{\begin{itemize}}
\newcommand{\ei}{\end{itemize}}
\newcommand{\bfi}{\begin{figure}}
\newcommand{\efi}{\end{figure}}
\newcommand{\ben}{\begin{enumerate}}
\newcommand{\een}{\end{enumerate}}
\newcommand{\bdes}{\begin{description}}
\newcommand{\edes}{\end{description}}
\newcommand{\bay}{\begin{array}}
\newcommand{\eay}{\end{array}}

\def\Ver{1}
\def\LongVer{0}
\newcommand{\be}{\begin{eqnarray}}
\newcommand{\ee}{\end{eqnarray}}

\begin{document}

\thispagestyle{empty}

\begin{center}
{\Large {\bf A general spline representation for nonparametric and semiparametric density estimates using diffeomorphisms.}}

\vspace{.5in}

Ethan Anderes$^{1}$  and Marc Coram$^{2}$\\

\vspace{.25in}

${}^1$Department of Statistics, University of California, Davis, CA 95616,
USA\\e-mail: {\tt anderes@stat.ucdavis.edu}

\vspace{.25in}

${}^2$Health Research and Policy Department, Stanford University, Palo Alto, CA 94304,
USA

\end{center}

\begin{footnotetext}[1]
{Supported in part by National Science Foundation grant DMS-1007480}
\end{footnotetext}
\begin{footnotetext}[2]
{Supported in part by National Institute of Health grant 5UL1 RR02574404}
\end{footnotetext}

\newpage

\begin{abstract}
A theorem of McCann \cite{mcc:95} shows that for any two absolutely continuous probability measures on $\Bbb R^d$ there exists a monotone transformation sending one probability measure to the other. 
A consequence of this theorem, relevant to statistics,  is that density estimation can be recast in terms of transformations. In particular, one can fix any absolutely continuous probability measure, call it $\Bbb P$, and then reparameterize the whole class of absolutely continuous probability measures as monotone transformations from $\Bbb P$.
In this paper we utilize this reparameterization of densities, as monotone transformations from some $\Bbb P$, to construct semiparametric and nonparametric density estimates. We focus our attention on classes of transformations, developed in the image processing and computational anatomy literature,  which are smooth,  invertible and which have attractive computational properties. 
 The techniques developed for this class of transformations allow us to show that a penalized maximum likelihood estimate (PMLE) of a smooth transformation from $\Bbb P$ exists and has a finite dimensional characterization, similar to those results found in the spline literature. 
 These results are derived utilizing  an Euler-Lagrange characterization of the PMLE which also establishes a surprising connection to a generalization of Stein's lemma for characterizing the normal distribution.
\end{abstract}


\vspace{.2in}


\thispagestyle{empty}


\newpage

\section{Introduction}

\setcounter{page}{1}

A theorem of McCann \cite{mcc:95} shows that for any two absolutely continuous probability measures on $\Bbb R^d$ there exists a monotone transformation sending one probability measure to the other. 
A consequence of this theorem, relevant to statistics,  is that density estimation can be recast in terms of transformations. In particular, one can fix any absolutely continuous probability measure, call it $\Bbb P$, and then reparameterize the whole class of absolutely continuous probability measures as monotone transformations from $\Bbb P$.  The advantage of this new viewpoint is the flexibility in choosing the target measure $\Bbb P$ which can allow  prior information on the shape of true sampling distribution. 
 For example, if it is known that the data is nearly Gaussian then choosing a  Gaussian $\Bbb P$  along with a strong penalty on transformations that are far from the identity allows one to construct penalized maximum likelihood estimates which effectively shrink the resulting nonparametric estimate in the direction of the Gaussian target $\Bbb P$. Moreover, when there is no knowledge about the true sampling measure, one can simply choose any absolutely continuous $\Bbb P$ and still construct a completely nonparametric density estimate. 

In this paper we utilize this reparameterization of densities, as monotone transformations from some $\Bbb P$, to construct semiparametric and nonparametric density estimates. We focus our attention on classes transformations which are smooth,  invertible and which have attractive computational properties.  Formally, we model our data   $X_1,\ldots, X_n$ as being generated by some diffeomorphism $\phi\colon \Bbb R^d \rightarrow\Bbb R^d$ of some fixed but known probably distribution $\Bbb P$ on $\Bbb R^d$:
\begin{equation}
\label{model1} 
X_1,\ldots,X_n\overset{iid}\sim \Bbb P\circ \phi .
\end{equation} 
The notation $\Bbb P\circ \phi$ is taken to mean that the probability of $X\in A$ is given by   $ \Bbb P(\phi(A))$ where $\phi(A)=\{ \phi(x)\colon x\in A \}$. An important observation is that the model (\ref{model1}) implies that $\phi(X)\sim \Bbb P$. Therefore, one can imagine estimating $\phi$ by attempting to ``deform'' the data $X_1,\ldots,X_n$ by a transformation which satisfies
\[ \phi(X_1),\ldots,\phi(X_n)\overset{iid}\sim \Bbb P. \]
One of the main difficulties when working with such a model is the invertibility condition on the maps $\phi$. The nonlinearity of this condition in $\Bbb R^d$ when $d>1$ makes constructing rich classes and optimizing over such classes difficult. 
One of the early attempts at circumventing such difficulties, found in \cite{and:11}, utilized  the class of quasi-conformal maps to generate penalized maximum likelihood estimates of $\phi$. However, these tools were only developed for $\Bbb R^2$ with no clear generalization for higher dimension. 
In this paper we adapt the powerful tools developed by Grenander, Miller,  Younes, Trouv\'e and co-authors in the image processing and computational anatomy literature   (see \cite{you:10} and the references therein) and 
apply them to the estimation of $\phi$ from data $X_1,\ldots, X_n$. Indeed, these techniques allow us to show that a penalized maximum likelihood estimate of $\phi$ exists and that the solution has a finite dimensional characterization, similar to those results found in the spline literature (see \cite{wahba:90}). 

We start the paper in Section \ref{rich} with an overview of using the  dynamics of time varying vector field flows to generate rich classes of diffeomorphisms. Then in sections \ref{pmle} and \ref{EL} we define our penalized maximum likelihood estimate (PMLE) of $\phi$ and prove not  only existence, but also establish a  finite dimensional characterization which is key for numerically computing the resulting density estimate. In Section \ref{steinSection} we notice a surprising connection with the Euler-Lagrange equation for the PMLE of $\phi$ and a generalization of Stein's lemma for characterizing the normal distribution (see \cite{stein:04}). In sections \ref{npe} and \ref{spe} we give examples of our new density estimate, first as a nonparametric density estimate and second as a semiparametric density estimate where a finite dimensional model is used for the target probability measure $\Bbb P$. We finish the paper with an appendix which contains some technical details used for the proofs of the existence of the PMLE and for the finite dimensional characterization.

\section{A rich class of diffeomorphisms}
\label{rich}

 In this section we give a brief overview of the technique of using the  dynamics of time varying vector field flows to generate rich classes of diffeomorphisms. These time varying flows have been utilized  with spectacular success in the fields of computational anatomy and image processing
 (see  \cite{alla:07, 
 Beg:2006ly,  
 conf/miccai/BegMTY03,
  Beg:2005qf,  
   cao:05, 
   dup:98, 
 Grenander:1998:CAE:309082.309089,
 Miller:1999bh, 
 Miller:2006zr, 
 Miller:2001ve, 
 Trouve:1998ys,  
 ty:dq,  
 vaillant:04, 
you:10,
Younes:2008}, and references therein).
  The tools developed in this body of work have been shown to be very powerful  for developing algorithms which optimize a diverse range of objective functions defined over classes of diffeomorphism.  It is these tools that we adapt for density estimation and statistics.

A map $\phi\colon \Omega \rightarrow\Bbb R^d$ is said to be a $C^k(\Omega,\Bbb R^d)$ diffeomorphism of the open set $\Omega \subset \Bbb R^d$ if  $\phi$ is one-to-one, maps onto $\Omega$, and  $\phi, \phi^{-1}\in C^{k}(\Omega,\Bbb R^d)$.  In what follows we generate classes of diffeomorphisms by time varying vector field flows. In particular, let  $\{ v_t\}_{t\in [0,1]}$  be a time varying vector field in $\Bbb R^d$, where $t$ denotes  `time' so that for each $t$, $v_t$ is a function mapping $\Omega$ into $\Bbb R^d$. Under mild smoothness conditions there exists a unique class of diffeomorphisms of $\Omega$, denoted $\{ \phi_t^v\}_{t\in [0,1]}$, which satisfy the following ordinary differential equation
\begin{align}
\label{ood}
 \partial_t \phi_t^v(x) &= v_t(\phi_t^v(x))
 \end{align}
 with boundary condition $\phi_0^v(x)=x$, for all $x\in \Omega$ (see Theorem \ref{ExistFlow} below). The interpretation of these flows is that $\phi_t(x)$ represents the position of a particle at time $t$, which originated from location $x$ at time $t=0$, and flowed according the the instantaneous velocity given by $v_t$. It will be convenient to consider the diffeomorphism that maps time $t$ to some other time $s$, this will be denoted $\phi_{ts}^{ v}(x)\equiv \phi_{s}^{v}( {\phi_{t}^{ v}}^{-1}(x)) $.

For the remainder of the paper we will assume that at each time $t$, $v_t$ will be a member of a  Hilbert space of vector fields mapping  $\Omega$ into $\Bbb R^d$ with inner product denoted by $\langle\cdot, \cdot \rangle_V$ and norm by $\|  \cdot \|_V$. 
Indeed, how one chooses the Hilbert space $V$ will determine the smoothness properties of the resulting class of deformations $\{\phi_t\}_{t\in[0,1]}$.  Once the Hilbert space $V$ is fixed we can define the following set of time varying vector fields.

\begin{definition}  
\label{defV01}
 Let  $V^{[0,1]}$ denote the space of measurable functions $v_t(x)\colon [0,1]\times \Omega\rightarrow \Bbb R^d$ such that $v_t\in V$ for all $t\in [0,1]$ and $\int_0^1 \| v_t\|^2_V dt <\infty $. 
\end{definition}
One clear advantage of this class  is that it can be endowed with a Hilbert space inner product if $V$ is a Hilbert space. Indeed,
 $V^{[0,1]}$ is a Hilbert space with inner product  defined by $\langle v,h \rangle_{V^{[0,1]}} \equiv\int_0^1 \langle  v_t, h_t\rangle_V dt$ (see Proposition \ref{Hspace} in the Appendix or Proposition 8.17 in \cite{you:10}).
For the remainder of the paper we typically use $v$ or $w$ to denote  elements of  $ V^{[0,1]}$ and  $v_t$ or $w_t$ to denote the corresponding elements of $V$ at any fixed time $t$.
An important theorem found in \cite{you:10} relates the smoothness of $V$ to the smoothness of the resulting diffeomorphism. 
Before we state the theorem, some definitions will be prudent. The Hilbert space $V$ is said to be continuously embedded in another normed space $H$ (denoted $V\hookrightarrow H$) if $V\subset H$ and there exists a constant $c$ such that
\[ \| v \|_H \leq c \| v\|_V \]
for all $v\in V$ where $\|\cdot \|_H$ denotes the norm in $H$. 
 Also we let $C_0^k(\Omega,\Bbb R^d)$ denote the subset of $C^k(\Omega,\Bbb R^d)$ functions whose partial derivatives of order $k$ or less all have continuous extensions to zero at the boundary $\partial \Omega$.
\begin{theorem}[\cite{you:10}, \cite{dup:98}]
\label{ExistFlow}
If $V\hookrightarrow C_0^k(\Omega,\Bbb R^d)$, then for any $v\in V^{[0,1]}$ there exists a unique class of $C^k(\Omega,\Bbb R^d)$ diffeomorphisms $\{\phi^v_t\}_{t\in[0,1]}$ which satisfy (\ref{ood}) and  $\phi_0^v(x) = x$ for all $x\in \Omega$. 
\end{theorem}

To derive our finite dimensional characterization we will make the additional assumption that $V$ is a reproducing kernel Hilbert space of vector fields. This will guarantee the existence of a reproducing kernel
 $K( x,y)\colon \Omega \times \Omega \rightarrow \Bbb R^{n\times n}$ which can be used to compute the evaluation functional. In particular, $K$ has the property that for any $x\in \Omega$ and $f\in V$ the following identity holds $\langle  K(x,\cdot)  p, f\rangle _V=  f(x)\cdot p$ for all $p\in \Bbb R^d$.
 To simplify the following computations we will only work with kernels of the form $K(x,y) = R(x,y)I_{d\times d}$ where $R:\Omega\times\Omega\rightarrow \Bbb R$ is a positive definite function and $I_{d\times d}$ is the $d$-by-$d$ identity matrix. 
 
 Now the class of diffeomorphisms we consider in this paper  corresponds to the set of all time varying vector field flows evaluated at $t=1$: $\phi^v_1$, where $v$ ranges through $V^{[0,1]}$. The class $V^{[0,1]}$ will be completely specified by the reproducing kernel $R(x,y)I_{d\times d}$ which has the flexibility to control the smoothness of the resulting maps $\phi_1^v$ through Theorem \ref{ExistFlow}.
%
\section{Penalized maximum likelihood estimation}
\label{pmle}

In this section we construct a penalized maximum likelihood estimate (PMLE) of $\phi^v_1$ given the data $X_1,\ldots, X_n\overset{iid}\sim \Bbb P\circ \phi^v_1$ where $\{\phi^v_t\}_{t\in[0,1]}$ satisfies (\ref{ood}) and $v\in V^{[0,1]}$.  Under mild assumptions on $V$ and the density of $\Bbb P$ we prove the existence of a PMLE estimate of $\hat v$, whereby obtaining an estimate $ \Bbb P\circ \phi^{\hat v}_1$ of the true sampling  distribution.

The target probability measure $\Bbb P$ is assumed known with a bounded density with respect to Lebesque measure on $\Bbb R^d$. 
Therefore, by writing the density of $\Bbb P$ as $\exp H$ for some  function $H\colon\mathbb R^d\rightarrow \mathbb R\cup \{-\infty\}$, the probability measure $\Bbb P\circ \phi^v_1$ has density given by
 \[ d\,\Bbb P\circ \phi_1^v(x)  =\det (D\phi_1^v(x))  \exp H \circ \phi_1^v(x) dx  \]
 where $\det (D\phi^v_1(x))$ is defined as the  determinant of the Jacobian of $\phi^v_1$ evaluated at $x\in \Omega$ (always positive by the orientation preserving nature of $\phi^v_1$). 
Since $\phi_1^v$ ranges over an infinite dimensional space of diffeomorphisms, the  likelihood for $v$ given the data  will typically be unbounded as $v$ ranges in $V^{[0,1]}$. The natural solution is to regularize the log likelihood using the corresponding  Hilbert space norm on $V^{[0,1]}$ with a  multiplicative tuning factor $\lambda/2$. 
The penalized log-likelihood (scaled by $1/n$) for the unknown vector field $v$ flow given data $X_1,\ldots, X_n\overset{iid}\sim \Bbb P\circ \phi^v_1$ is then given by
 \begin{equation}
 \label{energy1}
 E_\lambda(v)\equiv  \frac{1}{n} \sum_{k=1}^n \log\text{det}[D\phi(X_k) ]+H\circ\phi(X_k)  - \frac{\lambda}{2}  \int_0^1 \| v_t \|_V^2 dt.
 \end{equation}
 The estimated vector field $\hat v$ is chosen to be any element of $V^{[0,1]}$ which maximizes $E_\lambda$ over $V^{[0,1]}$.
 The following theorem establishes that such a $\hat v$ exists.
\begin{claim}
\label{claim1}
Let $V$ be a  Hilbert space which is  continuously  embedded in $C_0^2(\Omega, \Bbb R^d)$ where $\Omega$ is a bounded open subset of $\Bbb R^d$. Suppose $e^{H(\cdot)}$ is a bounded and continuous density  on $\Omega$. Then there exists a time varying vector field $\hat v \in V^{[0,1]}$ such that 
\begin{equation}
\label{exist}
E_\lambda(\hat v)=\sup_{v\in V^{[0,1]}}E_\lambda(v).
\end{equation}
\end{claim}
\begin{proof}

We first establish $\sup_{v\in V^{[0,1]}}E_\lambda(v)<\infty$ by splitting the energy $E_\lambda$ into three parts
\begin{equation}
\label{decomp}
E_\lambda(v)=\underbrace{\frac{1}{n}\sum_{k=1}^n   \log\text{det} D\phi_1^v(X_k) }_{:=E_1(v)}+ \underbrace{\frac{1}{n}\sum_{k=1}^n H\circ\phi_1^v(X_k)}_{:=E_2(v)}  \underbrace{-\frac{\lambda}{2} \int_0^1 \| v_t \|_V^2 dt.}_{:=E_3(v)}  
\end{equation}
Notice that each term is well defined and finite whenever $v\in V^{[0,1]}$, since
the assumption $V\hookrightarrow C_0^2(\Omega, \Bbb R^d)$ is sufficient for  Theorem 8.7   in \cite{you:10} to apply to the class  $ V^{[0,1]}$. In particular, for any  $v\in V^{[0,1]}$  there exists a unique class of $C^1$ diffeomorphisms of $\Omega$, $\{\phi_t^v\}_{t\in [0,1]}$, which satisfies (\ref{ood}) (also see Theorem 2.5 in \cite{dup:98}). 
The term $E_2(v)$ is clearly bounded from above since  $\sup_{x\in\Omega}H(x)<\infty$ by assumption.
For the remaining two terms notice that  the determinant of the Jacobian is given by   $ \log\det D\phi_1^v(x)  = \int_0^1  \text{div}\,  v_t (\phi_t^v(x))dt$ (by equation (\ref{div}) in the Appendix).
Therefore
\begin{align}
E_1(v) + E_3(v) & = \frac{1}{n}\sum_{k=1}^n   \int_0^1 \left(  \text{div}\,  v_t (\phi_t^v(X_k))  - \frac{\lambda}{2}   \| v_t \|_V^2
\right)dt \nonumber\\
 & \leq \frac{1}{n}\sum_{k=1}^n   \int_0^1 \left( \sup_{x\in \Omega} |\text{div}\,  v_t (x)|  - \frac{\lambda}{2}   \| v_t \|_V^2
\right)dt \nonumber \\
 & \leq    \int_0^1 \left( c\| v_t \|_V  - \frac{\lambda}{2}   \| v_t \|_V^2
\right)dt, \,\,\text{by the assumption $V\hookrightarrow C_0^2(\Omega, \Bbb R^d)$} \nonumber \\
& \leq \frac{c^2}{2\lambda}<\infty. \nonumber
\end{align}

Now let $v^1,v^2,\ldots$ be any maximizing sequence that satisfies $\lim_{m\rightarrow \infty} E(v^m) = \sup_{v\in V^{[0,1]}}E_\lambda(v)$. Since $\sup_{v\in V^{[0,1]}}E_\lambda(v)<\infty $  we can construct the sequence $v^m$ so that there exists an $M<\infty$ such that $\| v^m \|_{V^{[0,1]}}\leq M$ for all $m$. Since $\Omega$ is bounded, closed finite balls in $V^{[0,1]}=L^2([0,1],V)$ are weakly compact (by \cite{dup:98}). Therefore we may  extract a subsequence from $v^m$ (relabeled by $m$) which weakly converges to  a $\hat v \in V^{[0,1]}$. In particular, $\langle v^m, w \rangle_{V^{[0,1]}}\rightarrow \langle \hat v, w \rangle_{V^{[0,1]}}$ for all $w\in V^{[0,1]}$. Furthermore we have lower semicontinuity of the norm
\begin{equation}
\label{liminf}
 \liminf_{m\rightarrow \infty} \| v^m \|^2_{V^{[0,1]}} \geq  \| \hat v \|^2_{V^{[0,1]}}.
\end{equation}
Now by Theorem 3.1 in \cite{dup:98} 
\if\Ver\LongVer{ 
{\flushleft\textcolor{blue}{$\downarrow$---------begin long version---------}}\newline
\cite{dup:98}  applies since the assumption $[W_0^{3,2}(\Omega)]^3$ (where $\Omega\subset \Bbb R^3$ in their paper) is only used to establish that $[W_0^{3,2}(G)]^3\hookrightarrow C_0^1(G)$. Since we are assuming $V\hookrightarrow C_0^2(\Omega, \Bbb R^d)\hookrightarrow C_0^1(\Omega, \Bbb R^d)$ we are free to use the results in \cite{dup:98} . 
{\flushleft\textcolor{blue}{$\uparrow$------------end long version---------}}\newline
} \fi
 we have that  $\phi^{v^m}_t(x) \rightarrow \phi^{\hat v}_t(x)$ uniformly in $t\in [0,1]$ as $m\rightarrow \infty$. 
This allows us to show that $\log \det D\phi_1^{v^m}(x)\overset{m\rightarrow\infty}\longrightarrow \log \det D\phi_1^{\hat v}(x)$ for every $x\in \Omega$. To see why, one can use similar reasoning as in \cite{cao:05}. First write
\begin{align*}
|\log \det D\phi_1^{v^m}(x)- \log \det D\phi_1^{\hat v}(x)| &= \left| \int_{0}^1 \text{div}\, v_t^m(\phi_t^{v^m}(x)) -  \text{div}\, \hat v_t(\phi_t^{\hat v}(x)) dt \right| 
=I + I\!I 
\end{align*}
where the first term $I$ satisfies 
\begin{align*}
I &\equiv \left| \int_{0}^1 \text{div}\, v^m_t(\phi_t^{v^m}(x)) -  \text{div}\, v^m_t(\phi_t^{\hat v}(x)) dt \right| \\
&\leq \int_{0}^1 \|\text{div}\, v^m_t\|_{1,\infty} \bigl|\phi_t^{v^m}(x) -  \phi_t^{\hat v}(x)\bigr| dt  \\
&\leq   \int_{0}^1 c\| v^m_t\|_{V} \bigl|\phi_t^{v^m}(x) -  \phi_t^{\hat v}(x)\bigr| dt,\,\,\text{since $V\hookrightarrow C_0^2(\Omega,\Bbb R^d)$} \\
&\leq  c\|  v^m\|_{V^{[0,1]}} \Bigl[\int_{0}^1 \underbrace{\bigl|\phi_t^{v^m}(x) -  \phi_t^{\hat v}(x)\bigr|^2}_\text{ $= o(1)$ uniformly in $t$} dt\Bigr]^{1/2},\,\,\text{by H\"older.} \\
&\rightarrow 0,\,\text{ since $\| v^m\|_{V^{[0,1]}}\leq M$ for all $m$.}
\end{align*}
For the second term $I\!I$ notice that  the map sending $v\mapsto \int_0^1 \text{div}\, v_t(y_t) dt$ is a bounded linear functional on $V^{[0,1]}$ (using the fact that $V\hookrightarrow C_0^1(\Omega, \Bbb R^d)$) where $y_t \equiv \phi_t^{\hat v}(x)$. By the Riesz representation theorem there exists a $w^{\hat v}\in V^{[0,1]}$ such that $\int_0^1 \text{div}\, v_t(y_t) dt = \langle v,w^{\hat v} \rangle_\text{\tiny $V^{[0,1]}$}$. Therefore  
 \begin{align*} 
I\!I 
&\equiv \left| \int_{0}^1 \text{div}\, v_t^m(\phi_t^{\hat v}(x)) -  \text{div}\, \hat v_t(\phi_t^{\hat v}(x)) dt \right|\\
&=   \Bigl| \langle v^m -  \hat v, w^{\hat v}\rangle_{V^{[0,1]}}\Bigr|\rightarrow 0, \,\,\text{by weak convergence.}
\end{align*}
Combining the results for $I$ and $I\!I$ we can conclude that $\log \det D\phi_1^{v^m}(x)\overset{m\rightarrow\infty}\longrightarrow \log \det D\phi_1^{\hat v}(x)$ for every $x\in \Omega$.

To finish the proof notice that
\begin{align}
\sup_{v\in V^{[0,1]}} E_\lambda(v) &= \lim_{m\rightarrow \infty} E(v^m)= \limsup_{m\rightarrow \infty} E(v^m)\nonumber\\
&= \frac{1}{n}\sum_{k=1}^n   \log\text{det} D\phi_1^{\hat v}(X_k)+ \frac{1}{n}\sum_{k=1}^n H\circ\phi_1^{\hat v}(X_k) -\frac{\lambda}{2} \liminf_{m\rightarrow \infty} \int_0^1 \| v^m_t \|_V^2 dt \nonumber\\
&\leq \frac{1}{n}\sum_{k=1}^n   \log\text{det} D\phi_1^{\hat v}(X_k)+ \frac{1}{n}\sum_{k=1}^n H\circ\phi_1^{\hat v}(X_k) -\frac{\lambda}{2}  \int_0^1 \| \hat v_t \|_V^2 dt,\,\,\text{ by (\ref{liminf})}\nonumber \\
&= E_\lambda(\hat v) \nonumber
\end{align}

  \end{proof}

One of the important facts about any vector field flow $\hat v\in V^{[0,1]}$ which maximizes $E_\lambda$  is that the resulting estimated transformation  $\phi_1^{\hat v}$ is a geodesic (or minimum energy) flow with respect to the vector field norm $\int_0^1 \| \hat v_t\|^2_V dt$. To see this is first notice that the parameterization of time $t=1$ maps, $\phi_1^v$, by vector fields $v\in V^{[0,1]}$ is a many-to-one  parameterization. 
In other words there exist multiple pairs of vector fields $v,w\in V^{[0,1]}$ such that $\phi_1^v =\phi^w_1$ but $v\neq w$. Notice, however, that the log-likelihood term in $E_\lambda$ only depends on   $\phi_1^{\hat v}$. This implies that any maximizer $\hat v$ of $E_\lambda$ must simultaneously minimize  the penalty $\int_0^1 \| \hat v_t\|^2_V dt$ over the class of all $w\in V^{[0,1]}$ which has the same terminal value, i.e.\! $\phi^{\hat v}_1=\phi^w_1$.
 Consequently, the PMLE estimate $\hat v$ must be a geodesic flow. An important consequence  is that  geodesic flows $\{ \hat v_t \}_{t\in [0,1]}$ are completely determined by the initial vector field $\hat v_0$.  This will become particularly important in the next section where the initial velocity field will be completely parameterized by $n$ coefficient vectors.

%
%
%
%
\section{Spline representation from Euler-Lagrange}
\label{EL}

In this section we work under the additional assumption that $V$ is a reproducing kernel Hilbert space. This assumption allows one to derive the Euler-Lagrange equation for any maximizer $\hat v$ of which satisfies (\ref{exist}). This leads to a finite dimensional characterization of  $\hat v$  which  parallel those results found in the spline literature for function estimation.

\begin{claim}
\label{claim2}
Let $V$ be a  reproducing kernel Hilbert space, with  kernel $R(x,y)I_{d\times d}$, continuously  embedded in $C_0^3(\Omega, \Bbb R^d)$  where $\Omega$ is bounded open subset of $\Bbb R^d$. Suppose $e^{H(\cdot)}$ is a $C^1(\bar \Omega,\Bbb R)$ density  on $\Omega$. Then any time varying vector field $\hat v \in V^{[0,1]}$ which satisfies (\ref{exist}) also satisfies the following
 Euler-Lagrange equation: 
 \begin{align}
 \label{ELeq}
 \hat v_t(x)&=  \frac{1}{\lambda n}\sum_{k=1}^n \beta_{k,t}^T R(x,X_{k,t})  +  \frac{1}{\lambda n}\sum_{k=1}^n   \nabla_{y} R(x,y)\Bigr|_{y= X_{k,t}}
\end{align}
where $X_{k,t}\equiv \phi_t^{\hat v} (X_k)$ and 
$ \beta_{k,t}\equiv   \nabla H(X_{k,1}) D\phi^{\hat v}_{t1}(X_{k,t})  +\nabla \log\det D\phi^{\hat v}_{t1} (X_{k,t})$. 
\end{claim}
\begin{proof} 
Let $E_1, E_2$ and $E_3$ decompose $E_\lambda$ as in (\ref{decomp}). Notice first that if  $h\in V^{[0,1]}$  and $\epsilon \in \Bbb R$ then $2 E_3(\hat v+\epsilon h) = {\lambda}\| \hat v \|^2_{V^{[0,1]}} + \epsilon 2 \lambda\langle v,h \rangle_{V^{[0,1]}}  +\epsilon^2{\lambda}\| h \|^2_{V^{[0,1]}} $. Therefore $E_3(\hat v+\epsilon h)$ is differentiable with respect to $\epsilon$ with derivative given by
\begin{equation}
 \label{dderE1}
{\partial_\epsilon}  E_3(\hat v+\epsilon h)\bigr|_{\epsilon=0}
=  \int_0^1 \langle h_t,\lambda \hat v_t \rangle_V dt.
\end{equation}
In addition, Theorem 8.10 of \cite{you:10}
 implies that  $\phi^{\hat v+\epsilon h}_1(x)$  is differentiable at $\epsilon = 0$. 
Now, the assumption  $H\in C^1(\bar\Omega)$ combined with equation (\ref{111}), in the Appendix, gives 
\begin{align} 
\partial_\epsilon E_2(\hat v+\epsilon h)\bigr|_{\epsilon=0} \nonumber
&= -\frac{1}{n}\sum_{k=1}^n  \nabla H(\phi_1^{\hat v }(X_k)) \cdot {\partial_\epsilon}  \phi^{\hat v+\epsilon h}_{1}(X_k)\bigr|_{\epsilon = 0} \nonumber \\
&= -\frac{1}{n}\sum_{k=1}^n  \nabla H(\phi_1^{\hat v }(X_k)) \cdot \int_0^1  \bigl\{D\phi^{ \hat v}_{u1} h_u \bigr\}\circ{\phi^{ \hat v}_{u}(X_k)}\,   du \nonumber \\
&= -\frac{1}{n}\sum_{k=1}^n \int_0^1  \bigl\{  \nabla H( X_{k,1}) D\phi^{\hat v}_{u1}(X_{k,u}) \bigr\} \cdot  h_u (X_{k,u})\,   du \nonumber \\
&=  \int_0^1\Bigl\langle  h_u(\cdot),  -\frac{1}{n}\sum_{k=1}^n \bigl\{ \nabla H( X_{k,1}) D\phi^{\hat v}_{u1}(X_{k,u}) \bigr\}^T R(\cdot,X_{k,u})\Bigr\rangle_V \,du  \label{derE2}
\end{align}
Finally, Proposition \ref{proo}, from the Appendix, implies $E_3(\hat v+\epsilon h )$ is differentiable at $\epsilon = 0$ with derivative given by  
\begin{align}
\partial_\epsilon  E_1(\hat v+\epsilon h)\bigr|_{\epsilon=0} \nonumber
&  =- \frac{1}{n} \sum_{k=1}^n  \partial_\epsilon  \log \det D\phi_{1}^{\hat v+\epsilon h }(X_k)\bigr|_{\epsilon = 0} \nonumber\\
&=- \frac{1}{n} \sum_{k=1}^n   
 \int_0^1  \Bigl[ h_u \cdot \nabla \log\det D\phi_{u1}^{\hat v}  + \text{\rm div}\, h_u\Bigr] \circ  \phi^{\hat v}_{u}(X_k)  \, du \nonumber\\
&= \int_0^1 \Bigl\langle h_u(\cdot), - \frac{1}{n} \sum_{k=1}^n \left\{ \nabla \log\det D\phi_{u1}^{\hat v}(X_{k,u})\right\}^T  R(\cdot,X_{k,u})+\nabla_{y} R(\cdot,y)\bigr|_{y=X_{k,u}}\Bigr\rangle_V \,du \label{derE3} 
\end{align}
{\em Remark:} the above equation requires $\partial_{x_i } ( e_i\cdot h_u(x))  = \partial_{x_i}   \langle e_i R(\cdot, x),  h_u\rangle_V=  \langle e_i  \partial_{x_i} R(\cdot, x),  h_u\rangle_V$ which follows since $\text{div}\, h_u \in V$ by the assumption $V\hookrightarrow C_0^2(\Omega,\Bbb R^d)$ (see \cite{aro:50}).
Now from (\ref{dderE1}),  (\ref{derE2}) and (\ref{derE3}), the energy $ E_\lambda(\hat v+\epsilon h)$ is differentiable with respect to $\epsilon$ at $0$ and
\begin{equation}
\label{EEEl}
0={\partial_\epsilon}  E_\lambda(\hat v+\epsilon h)\bigr|_{\epsilon=0}= \langle \mathcal E ^{\hat v},h \rangle_{V^{[0,1]}}
\end{equation}
where
\begin{equation}
\label{CalE}
 \mathcal E ^{\hat v}_t =\lambda \hat v_t   - \frac{1}{n}\sum_{k=1}^n \beta_{k,t}^T  R(\cdot,X_{k,t})  - \frac{1}{n}\sum_{k=1}^n \nabla_{y} R(\cdot,y)\bigr|_{y=X_{k,t}}
 \end{equation}
with $\beta_{k,t}\equiv  \nabla H( X_{k,1}) D\phi^{\hat v}_{t1}(X_{k,t})  + \nabla \log\det D\phi_{t1}^{\hat v}(X_{k,t}) $.
Since $h\in V^{[0,1]}$ was arbitrary, equation (\ref{EEEl}) implies $\mathcal E^{\hat v}=0$, which then gives (\ref{ELeq}). {\em Remark:} we are using the fact that the zero function in a reproducing kernel space is point-wise zero since the evaluation functionals are bounded. 
\end{proof}

   There are a few things things to note here. First, the Euler-Lagrange equation  (\ref{ELeq}) only implicitly characterizes $\hat v$ since it appears on both sides of the equality ($\beta_{k,t}$ and $X_{k,t}$ also  depend on $\hat v$). Regardless,  (\ref{ELeq}) is useful since it implies that $\hat v$ must lie within a known $n\times d$ dimensional sub-space of $V^{[0,1]}$. In particular, as discussed at the end of Section \ref{pmle}, the estimate $\{ \hat v_t\}_{t\in [0,1]}$ is completely characterized by it's value at time $t=0$, i.e.\! $\hat v_0$ (by the geodesic nature of $\hat v$). Restricting equation (\ref{ELeq}) to  $t=0$ one obtains
   \begin{align}
 \label{ELeq0}
 \hat v_0(x)&=  \frac{1}{\lambda n}\sum_{k=1}^n \beta_{k,0}^T R(x,X_{k})  +  \frac{1}{\lambda n}\sum_{k=1}^n   \nabla_{y} R(x,y)\Bigr|_{y= X_{k}}.
\end{align}
  Simply stated, $\hat v_0$ has a finite dimensional spline characterization with spline knots set at the observations $X_1,\ldots, X_n$. Therefore to recover $\{ \hat v_t \}_{t\in[0,1]}$ one simply needs to find the $n$ vectors $\beta_{1,0},\dots, \beta_{n,0}$ which satisfy the following fixed point equation
        \begin{equation}
 \label{InitialMom}
  \beta_{k,0}= \nabla H(\phi^{\hat v}_1(X_k))  D\phi^{\hat v}_{1}(X_k)+ \nabla \log\det D\phi^{\hat v}_{1} (X_k)
  \end{equation}
 for all $k=1,\ldots,n$.

\section{Connection to Stein's Method}
\label{steinSection}

The Euler-Lagrange equation given in (\ref{ELeq}) has a surprising connection with a generalization of Stein's lemma for characterizing the normal distribution (see \cite{stein:04}). The main connection is that the Euler Lagrange equation for the PMLE estimate $\hat v_t$, simplified at initial time $t=0$ and terminal time $t=1$, can be reinterpreted as an empirical version of a generalization of Stein's lemma. This is interesting in it's own right, however, the connection may also bear theoretical fruit for deriving asymptotic estimation bounds on the nonparametric and semiparametric estimates derived from $\hat v$. In this section we make this connection explicit  with the goal of of motivating and explaining the Euler-Lagrange equation for $\hat v$ derived above.

To relate $\hat v_t$ at $t=0$ with Stein's lemma, and more generally Stein's method for distributional approximation,  first notice that (\ref{InitialMom}) implies the coefficients $\beta_{k,0}$,  from the implicit equation  (\ref{ELeq0}) for  $\hat v$, satisfy
$\beta_{k,0}
 = \nabla \log \hat f(X_k)
$ where  $\hat f= e^{ H\circ \phi^{\hat v}_1}  |D\phi_1^{\hat v}| $ is the estimated density of $X$ using the pullback of the target measure with the estimated diffeomorphisms $\phi^{\hat v}_1$.
 Now by computing the inner product of  both sides of the Euler-Lagrange equation (\ref{ELeq0}) with any vector field $u\in V$ and applying the reproducing property of the kernel $R(\cdot,\cdot)$ one derives
 \begin{align}
\lambda \langle \hat v_0, u\rangle_V
&=\Bbb E_n\bigl\{ \nabla \log \hat f(X)\cdot  u(X) + \text{div}\, u(X)\bigr\}. \label{stein2} 
\end{align}
where $\Bbb E_n$ denotes  expectation with respect to the empirical measure generated by the data: $\frac{1}{n}\sum_{k=1}^n \delta_{X_k}$. 
To relate with Stein first let $\Bbb E$ denote expectation with respect to the population density $f =    e^{ H\circ \phi} |D\phi| $ given in our basic model (\ref{model1}). Notice that a generalization of Stein's lemma shows that if the densities $f$ and $\hat f$ give rise to the same probability measure then
\begin{equation}
\label{ste}
0= \Bbb E\bigl\{ \nabla \log \hat f(X) \cdot u(X) + \text{div}\, u(X)\bigr\} 
 \end{equation}
for all $u$ in a large class of test functions $\mathcal U$ (see Proposition 4 in \cite{stein:04}).
For example, a simple consequence of Lemma 2 in \cite{stein:81} implies that when $\hat f$ is the density of a  $d$ dimensional Gaussian  distribution $\mathcal N_d(\hat\mu,1)$  and $X\sim \mathcal N_d(\mu,1)$ then $\hat\mu = \mu$ implies $ \Bbb E\bigl\{ -  (X-\hat\mu) \cdot u(X) + \text{div}\, u(X) \bigr\} =0$ 
for any bounded function $u:\Bbb R^d \rightarrow \Bbb R^d$ with bounded gradient. Stein's method, on the other hand, generally refers to  a technique for  bounding the distance between two probability measures $f$ and $\hat f$ using bounds on departures from a characterizing equation, such as  (\ref{ste}) for example (see \cite{chen:05} for an exposition). The bounds typically take the form
\begin{align} \sup_{h\in \mathcal H} \left| \int ( h f - h \hat f ) \right| &\leq \sup_{u\in \mathcal U} \bigl| \Bbb E\bigl\{ \nabla \log \hat f(X)\cdot u(X) + \text{div}\, u(X)\bigr\} \bigr| \label{ssmethod} 
 \end{align}
where $\mathcal H$ and $\mathcal U$ are two class of functions related through a set of differential equations.
  In our case, applying a H\"older's inequality to the Euler-Lagrange equation (\ref{stein2}) gives a bound on  right hand side of (\ref{ssmethod}) in terms of a regularization measurement on the PMLE $\hat v$ and an empirical process error:
    \begin{align}
    \label{reg1}
     \sup_{u\in \mathcal U} \bigl| \Bbb E\bigl\{ \nabla \log \hat f(X) \cdot u(X) + \text{div}\, u(X)\bigr\} \bigr|   &\leq \underbrace{ \lambda \| \hat v_0\|_V   \sup_{u\in \mathcal U} \|  u \|_V}_\text{regularization at $t=0$} + \underbrace{\sup_{u\in \mathcal U} \bigl| (\Bbb E -\Bbb E_n) \nu_{\hat f,u}\bigr|}_\text{ empirical process error}
 \end{align}
 where $\nu_{\hat f,u} = \nabla \log \hat f(X) u(X) + \text{div}\, u(X)$.  
 This makes it clear that  theoretical control of the PMLE estimate $\hat v_t$ at time $t=0$, using the Euler-Lagrange equation characterization (\ref{ELeq0}), allows asymptotic control of the distance between the estimated density $\hat f$ and the true density $f$.
 
 At terminal time $t=1$, there is a similar connection with Stein's lemma. In contrast to time $t=0$, which quantifies the distance between the estimated  and population densities $\hat f$ and $f$,  time $t=1$ quantifies the distance between $\phi(X)$ (the target measure) with $\phi_1^{\hat v}(X)$ (the push forward of the true population distribution though the estimated map).  To make the connection, one follows  the same line of argument as above to  find that  for any $u\in V$
\begin{align}
\lambda \langle \hat v_1, u\rangle_V &=  {\Bbb E}^{\hat v}_n\bigl\{ \nabla H(X)\cdot u(X ) + \text{div}\, u(X)\bigr\} \label{stein}
\end{align}
where ${\Bbb E}^{\hat v}_n$ denotes  expectation with respect to the empirical measure $\frac{1}{n}\sum_{i=1}^n \delta_{\phi^{\hat v}_1(X_k)}$, which is simply the push forward of the empirical measure $\frac{1}{n}\sum_{i=1}^n \delta_{X_k}$ through the estimated map $\phi^{\hat v}_1$.
Now the analog to (\ref{reg1})  becomes
   \begin{align}
   \label{reg2}
    \sup_{u\in \mathcal U} \bigl| {\Bbb E^{\hat v}}\bigl\{ \nabla H(X)\cdot  u(X) + \text{div}\, u(X)\bigr\} \bigr|   &\leq \underbrace{ \lambda \| \hat v_1\|_V  \sup_{u\in \mathcal U} \|  u \|_V }_\text{regularization at $t=1$} +\, {\sup_{u\in \mathcal U} \bigl| (\Bbb E^{\hat v} -\Bbb E_n^{\hat v}) \gamma_{u}\bigr|}
 \end{align}
 where $\gamma_{u} = \nabla H(X) u(X) + \text{div}\, u(X)$ and $\Bbb E^{\hat v}$ denotes expectation with respect to the push forward of the population density $f =    e^{ H\circ \phi} |D\phi| $ though the estimated map $\phi_1^{\hat v}$.  
Since the target measure $\Bbb P$ is assumed to have density $e^H$, this bounds the distributional distance between $\phi^{\hat v}_1(X)$ and $\Bbb P$ when $X\sim \Bbb P\circ \phi$.


\section{Nonparametric example}
\label{npe}

In this section we utilize the finite dimensional characterization of the PMLE $\hat v$ at time $t=0$, given in (\ref{ELeq0}), to construct nonparametric density estimates of the form $\hat f= e^{ H\circ \phi^{\hat v}_1}  |D\phi_1^{\hat v}| $ from {\em iid} samples $X_1,\ldots, X_n$.
As was discussed in the introduction, so long as the target measure $\Bbb P$ is absolutely continuous, the assumption that   $X_1,\ldots, X_n \overset{iid}\sim \Bbb  P\circ \phi$ encompasses all absolutely continuous measures. Since the class of diffeomorphisms $\{\phi_1^v\colon v\in V^{[0,1]}\}$ is nonparametric, the estimate $\hat f= e^{ H\circ \phi^{\hat v}_1}  |D\phi_1^{\hat v}| $ is inherently nonparametric regardless of the choice of target probability measure $\Bbb P$ (with density $e^H$). In effect, the choice of target $\Bbb P$ specifies a shrinkage direction for the nonparametric estimate: larger values of $\lambda$ shrink $\hat f$ further toward the target $\Bbb P$.
In this section we illustrate the nonparametric nature of the density estimate $\hat f$, whereas the next section explores semiparametric estimation with parametric models on the target $\Bbb P$.     One  key feature of our methodology is the use of the Euler-Lagrange equation (\ref{ELeq}) as a stopping criterion for a gradient based optimization algorithm for constructing $\hat v$. In fact, to avoid computational challenges associated with generating geodesics  with initial velocities given by (\ref{ELeq0}), we consider a finite dimensional subclass of $V^{[0,1]}$ which have geodesics that are amenable to computation (and for which gradients are easy to compute). The key is that we use the Euler-Lagrange identity (\ref{ELeq}) to measure of the richness of the subclass, within the larger infinite dimensional Hilbert space $V^{[0,1]}$, whereby allowing a dynamic choice  of the approximating dimension for a target resolution level.

Claim~\ref{claim2} shows that the PMLE vector field $\hat{v}\in V^{[0,1]}$ obeys a parametric form determined up to the identification of the $n$ functions $t\mapsto \beta_{k,t}$ as $t$ ranges in $[0,1]$. Moreover, the whole path of coefficients $\beta_{k,t}$ is determined from the initial values $\beta_{k,0}$, by the geodesic nature of $\hat v$. 
In this way, we are free to optimize, over the vectors $\{\beta_{1,0}, \ldots, \beta_{n,0}\}\subset \Bbb R^d$ using equation (\ref{ELeq0}) 
 and are guaranteed that the global maximum, over the full infinite dimensional space $\{ \phi_1^v\colon v\in V^{[0,1]}\}$, has this form.  Unfortunately, deriving geodesic maps with this type of initial velocity field is challenging.  
To circumvent this difficulty we choose an approximating  subclass of vector fields at time $t=0$ which are parametrized by the selection of $N$ knots $\{\kappa_{1},\ldots, \kappa_N\} \subset \Omega$ and $N$ initial momentum row vectors $\{\eta_{1},\ldots, \eta_N\} \subset  \mathbb{R}^d$ and have the form: 
\begin{equation}
\label{knots}
 v_0(x)= \sum_{k=1}^N \eta_{k}^T R(x,\kappa_{k}).
 \end{equation}
The knots $\{ \kappa_1,\ldots, \kappa_N\}$ need not be located at the data points $\{X_1,\ldots, X_n \}$. Indeed, we will see that alternative configurations of knots can be numerically beneficial. The key point is that vector fields at time $t=0$, which satisfy (\ref{knots}), generate geodesics with respect to norm $\bigl[\int_0^1 \| v^t  \|^2_V dt \bigr]^{1/2}$ that are easy to compute. Moreover, the variational derivatives of the terminal map $\phi^v_1$ with respect to the initial $\eta$ coefficients and the knots $\kappa$ are easily computed when utilizing  similar techniques  as those developed in \cite{vaillant:04} and \cite{alla:07}. This enables efficient gradient based algorithms for optimizing the PMLE criterion over the class generated by (\ref{knots}).

\begin{figure}[t]
\centering
\includegraphics[height = 2.2in]{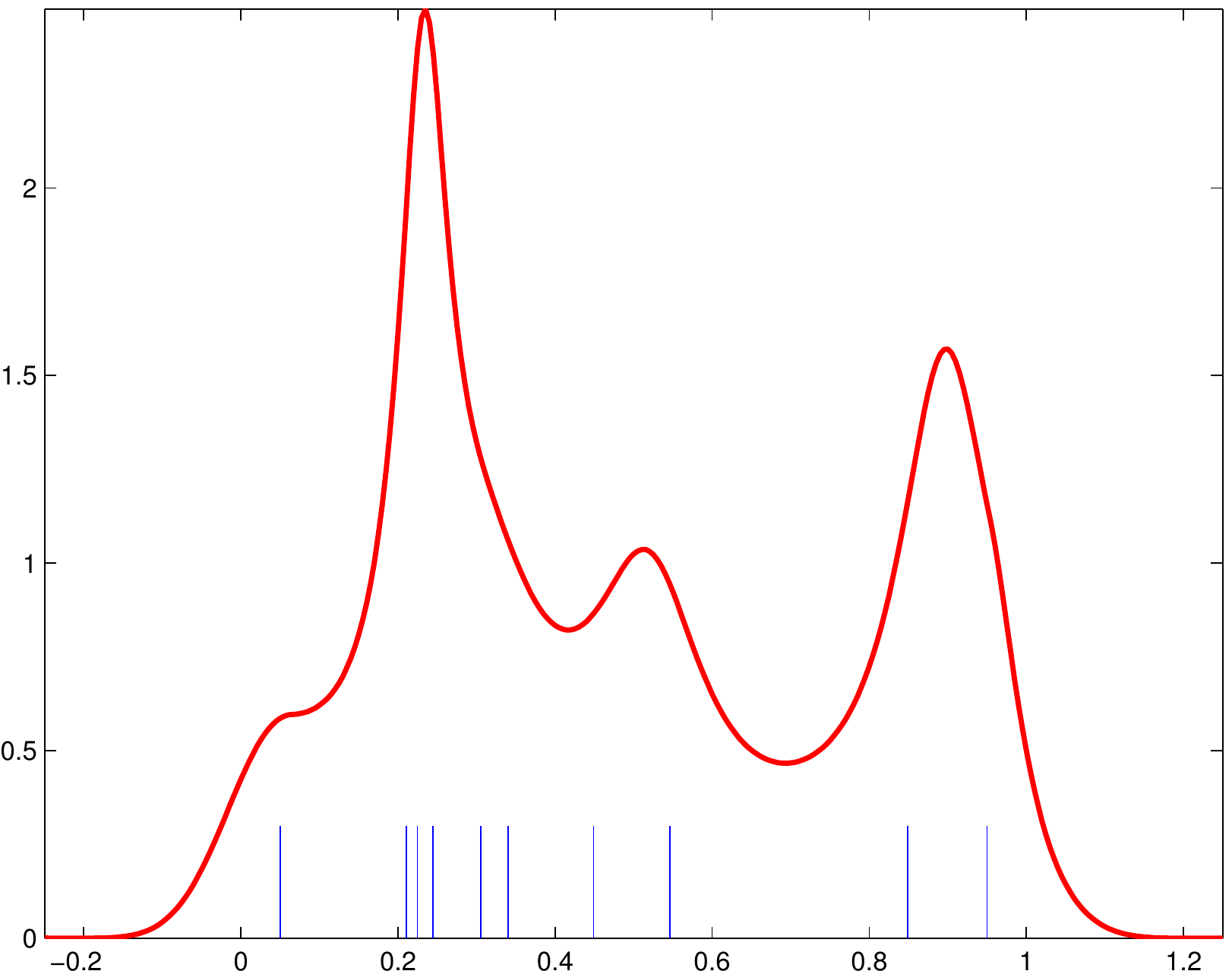}
\includegraphics[height = 2.2in]{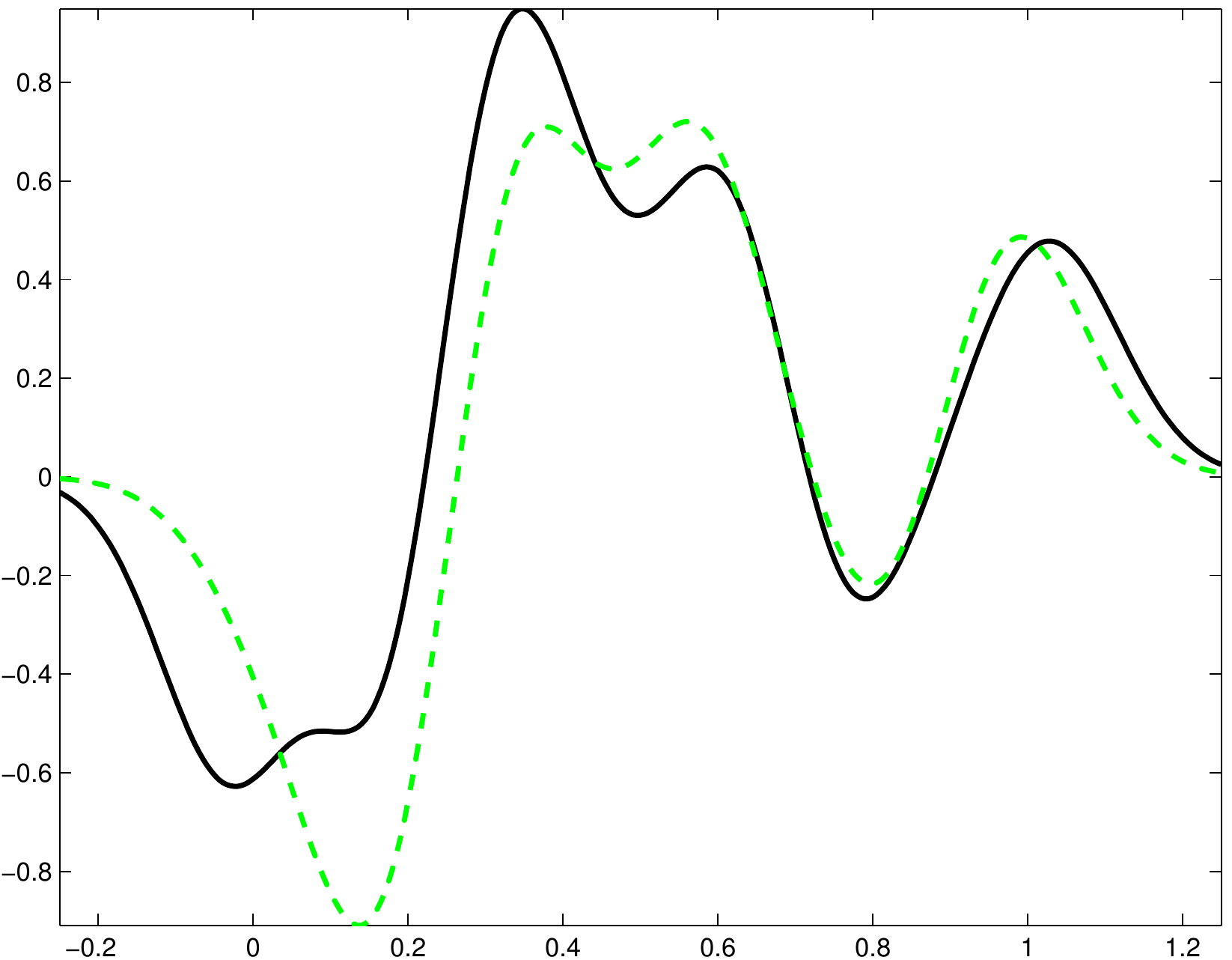}\\
\includegraphics[height = 2.2in]{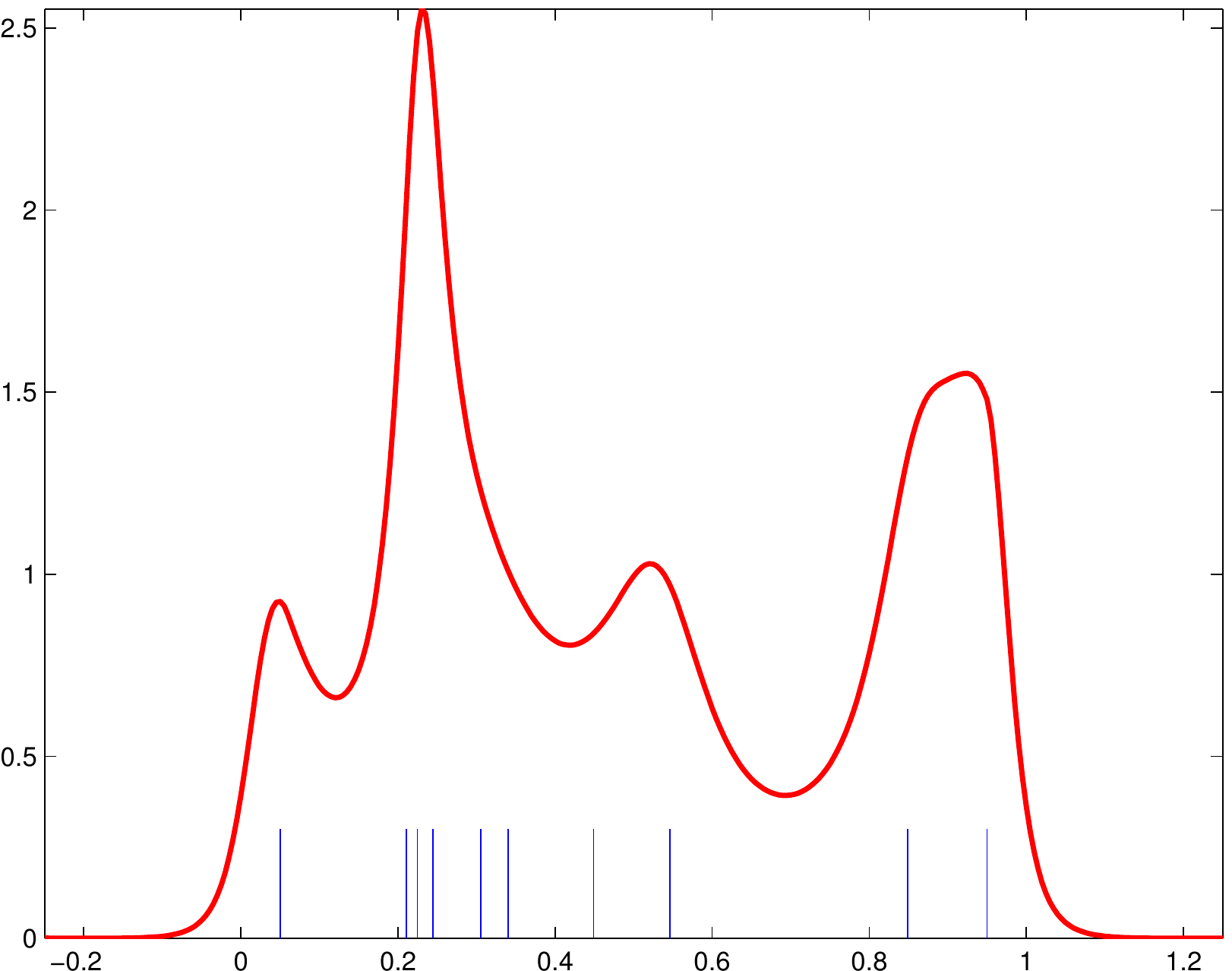}
\includegraphics[height = 2.2in]{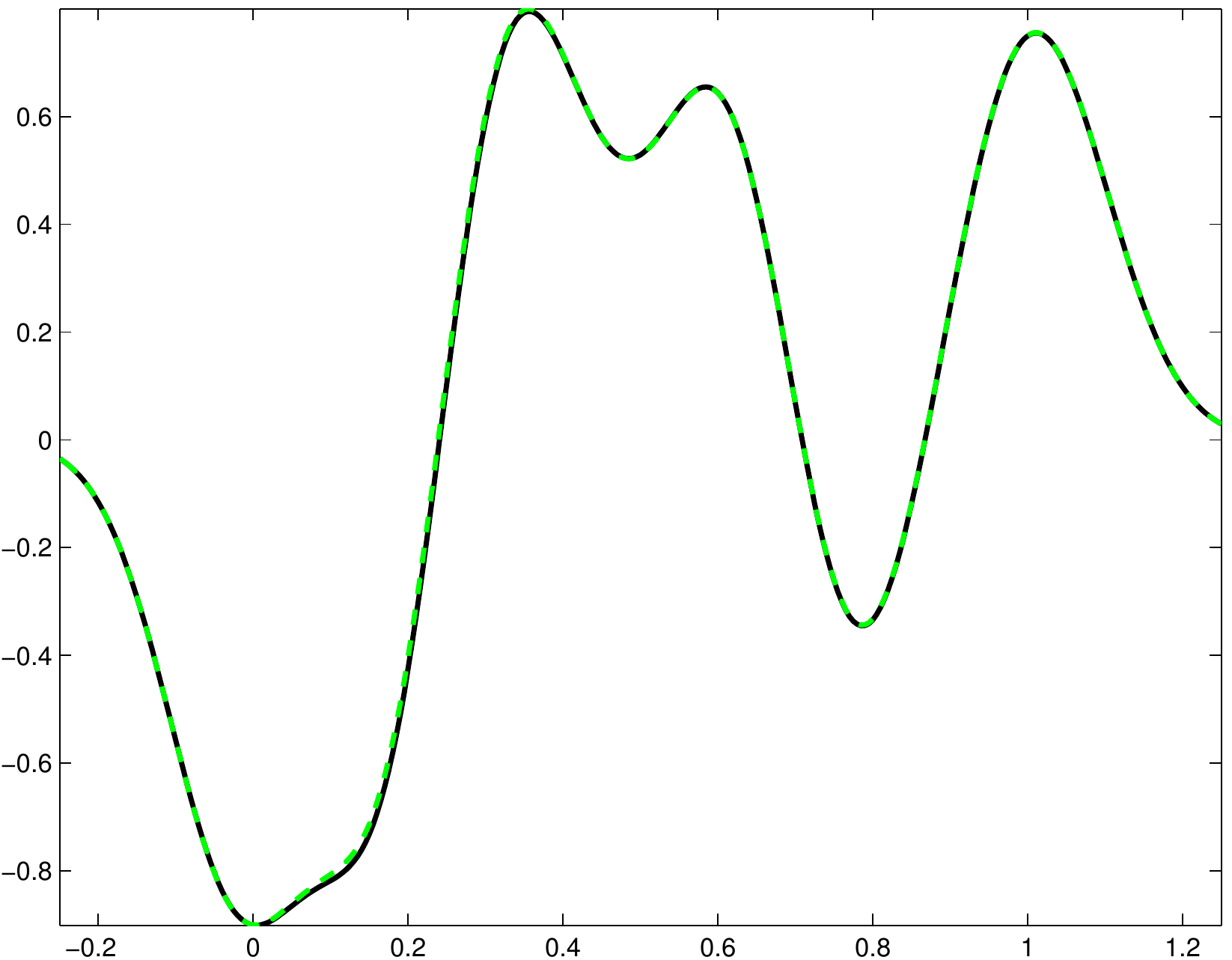}\\
\caption{In this example we compare two different knot configurations, in (\ref{knots}), for generating nonparametric density estimates using approximate solutions to the  Euler-Lagrange equation (\ref{ELeq}). The left column of images shows two different density estimates (red), based on the same data set (blue), using two different knot configurations (top-left uses $10$ knots, bottom-left uses $30$ knots). The right column of images show the corresponding  diagnostic curves which characterize the richness of the approximating subclass generated by the knots. The fact that the two diagnostic curves shown bottom-right are similar suggests that the $30$ knots used generate the approximating subclass by (\ref{knots}) is sufficiently rich to reach the stationary points of the penalized log likelihood $E_\lambda$ given in (\ref{energy1}).
See Section \ref{npe} for details.
 \label{f2}   
 }
\end{figure}

As a first illustration, we show that the na\"ive choice of initial knots obtained by setting $\{\kappa_1,\ldots, \kappa_N\}=\{ X_1,\ldots,X_n\}$ in (\ref{knots})  is {\em not} sufficient to solve  (\ref{ELeq}); then show how it can be easily fixed using the Euler-Lagrange methodology. 
 Our data set, shown with blue sticks in Figure~\ref{f2}, consists of $n=10$ independent samples from a mixture of two normals, truncated so the support is $[0,1]$.  Our target probability measure $\Bbb P$ is set to the uniform distribution on  $[0,1]$ (smoothly tapering to zero $0$ outside of $[0,1]$ for numerical convenience).
  For simplicity we choose the Gaussian kernel $R(x,y)=\exp\bigl(-\frac{(x-y)^2}{2\sigma^2}\bigr)$, with $\sigma=0.1$, to generate the RKHS $V$ and use the  penalty parameter $\lambda$ set to $10$. The top left plot in Figure~\ref{f2} shows the non-parametric density estimate $\hat f= e^{ H\circ \phi^{\hat v}_1}  |D\phi_1^{\hat v}| $ in red, generated by applying a gradient based optimization algorithm applied to the subclass (\ref{knots}) where the knots $\{\kappa_1,\ldots, \kappa_N\}=\{X_1,\ldots, X_n\}$ are kept fixed and the coefficients $\eta_1,\ldots, \eta_N$ are optimized by minimizing the penalized log likelihood function $E_\lambda (v)$ given in (\ref{energy1}).  
 To diagnose the richness of subclass (\ref{knots}) within the full Hilbert space we define the function $\mathcal D_t^v(x)$ for any $v\in V^{[0,1]}$ and any $t\in [0,1]$ as follows
 \begin{equation}
 \label{diag1}
 \mathcal D_t^v(x) \equiv    \frac{1}{ n}\sum_{k=1}^n \bigl[\beta_{k,t}^v \bigr]^T R(x, X_{k,t}^v )  +  \frac{1}{ n}\sum_{k=1}^n   \nabla_{y} R(x,y)\Bigr|_{y= X_{k,t}^v}
 \end{equation}
where $X_{k,t}^v\equiv \phi_t^{v} (X_k)$ and 
$ \beta_{k,t}^v\equiv   \nabla H(X^v_{k,1}) D\phi^{v}_{t1}(X_{k,t}^v)  +\nabla \log\det D\phi^{ v}_{t1} (X^v_{k,t})$. 
The function $\lambda v_t - \mathcal D_t^v$ serves as a diagnostic criterion in the sense that
  the Hilbert norm of $\lambda v_t - \mathcal D_t^v$ gives  the maximal rate of change of the penalized log-likelihood $E_\lambda(v)$, within the full infinite dimensional Hilbert space $V^{[0,1]}$. In particular, 
 \[   \biggl[\int_0^1\|\underbrace{ \lambda v_t-\mathcal D_t^v}_\text{\scriptsize diagnostic } \|_{V}^2dt\biggr]^{1/2} = \sup_\text{\small $\{ u\colon \| u \|_{V^{[0,1]}}= 1 \}$} \left[\frac{d}{d\epsilon} E_\lambda (v+\epsilon u)\right]_{\epsilon = 0} .   \]
 Therefore if   $\lambda v_t(x)-\mathcal D_t^v(x) = 0$ for all $t\in [0,1]$ and $x\in \Bbb R^d$, then $v$ satisfies the Euler-Lagrange equation. Discrepancies between $\lambda v_t(x)$ and $\mathcal D_t^v$ when optimizing over the subclass (\ref{knots}) indicates the subclass that is insufficient rich to reach the stationary points of $E_\lambda(v)$. 
 The diagnostic plots in this example, which correspond to our density estimate shown in the upper-left image of Figure~\ref{f2},  are shown in the upper-right plot of Figure~\ref{f2} where $\lambda v_0(x)$ is plotted in  black and $\mathcal D_0^v(x)$ is plotted as a dashed green line. The large amount of discrepancy between $\lambda v_0(x)$  and $\mathcal D_0^v(x)$ indicates that the knots $\{\kappa_1,\ldots, \kappa_N\}=\{ X_1,\ldots,X_n\}$  are insufficient.

\begin{figure}[t]
\centering
\includegraphics[height = 2.2in]{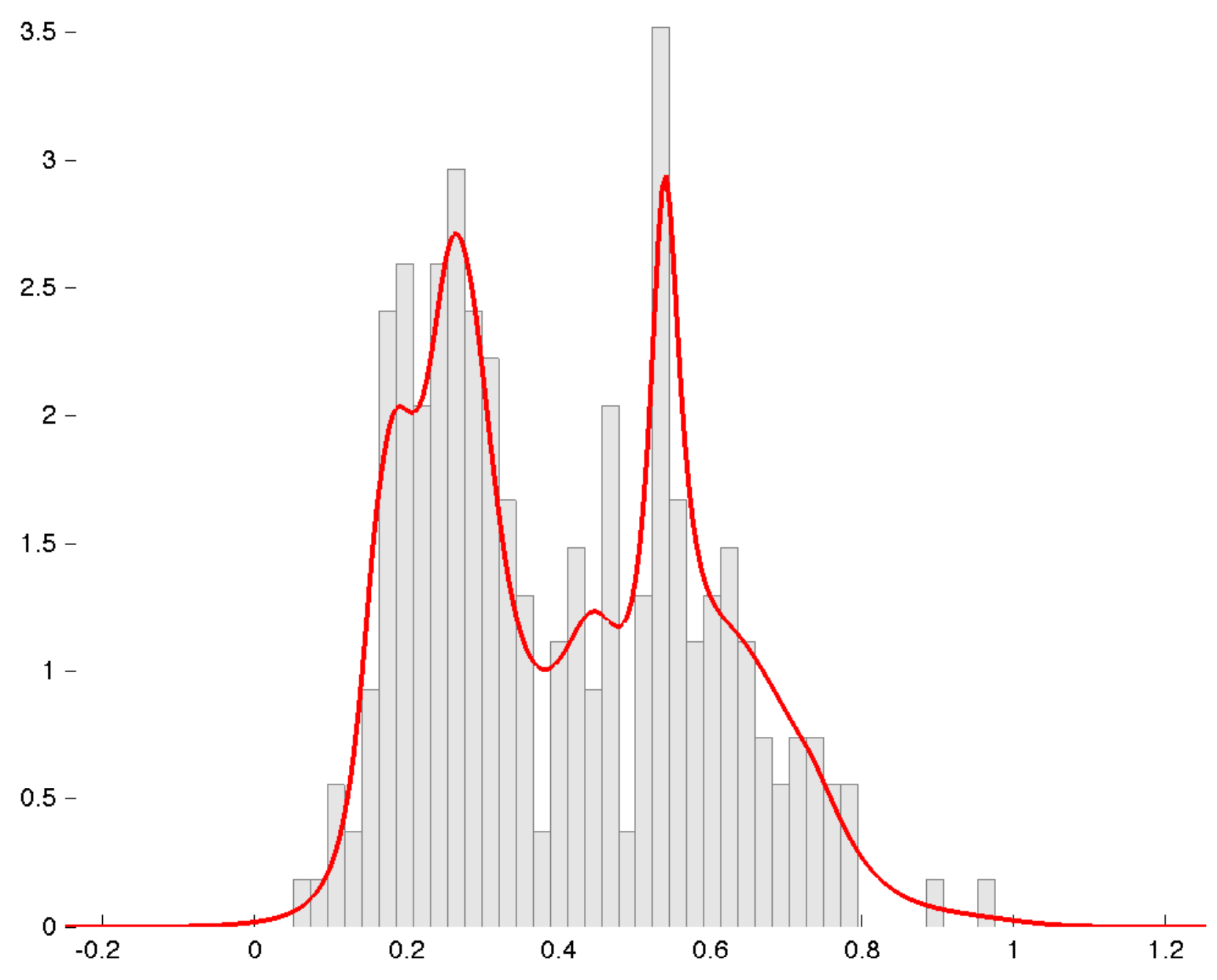}
\includegraphics[height = 2.2in]{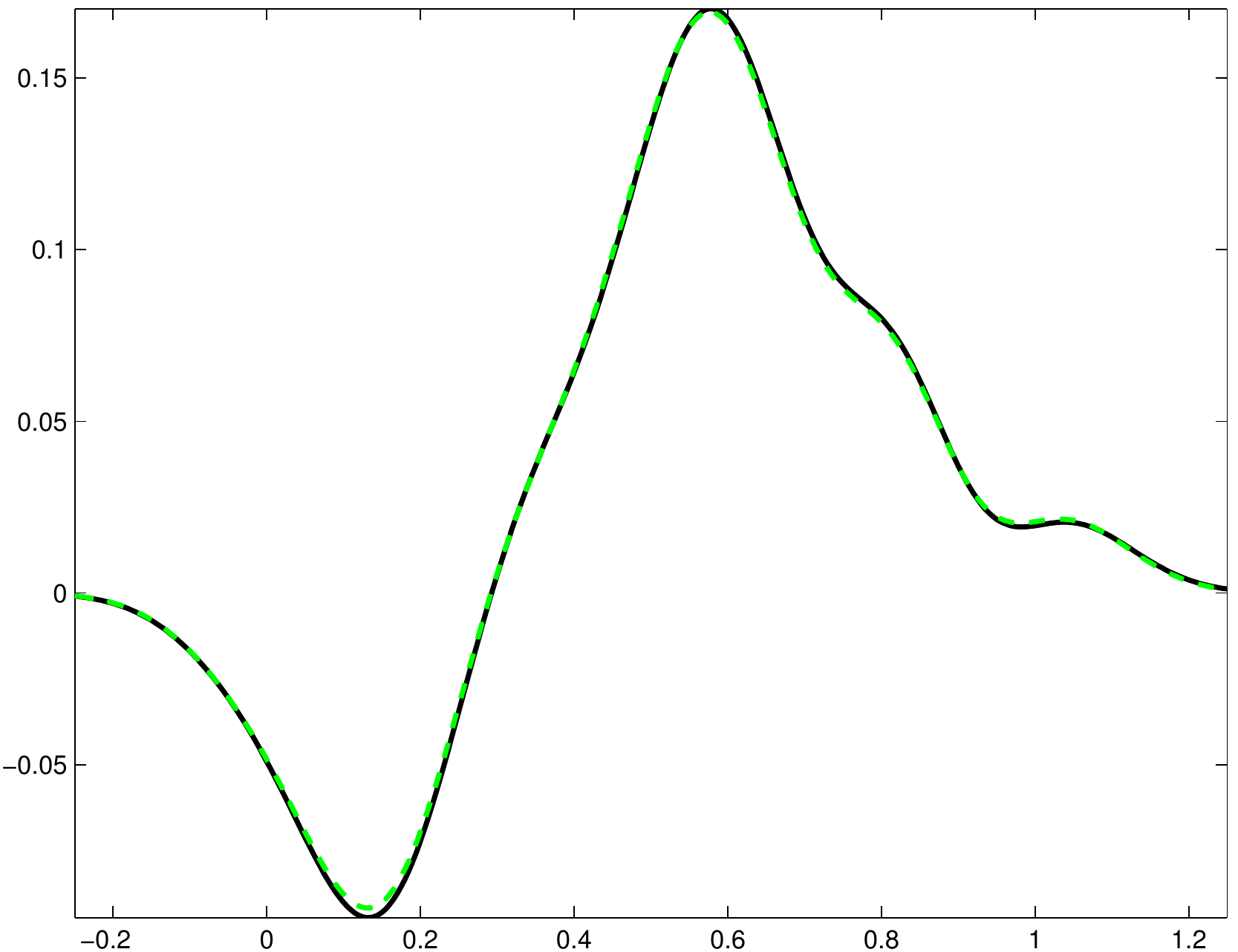}\\
\includegraphics[height = 2.2in]{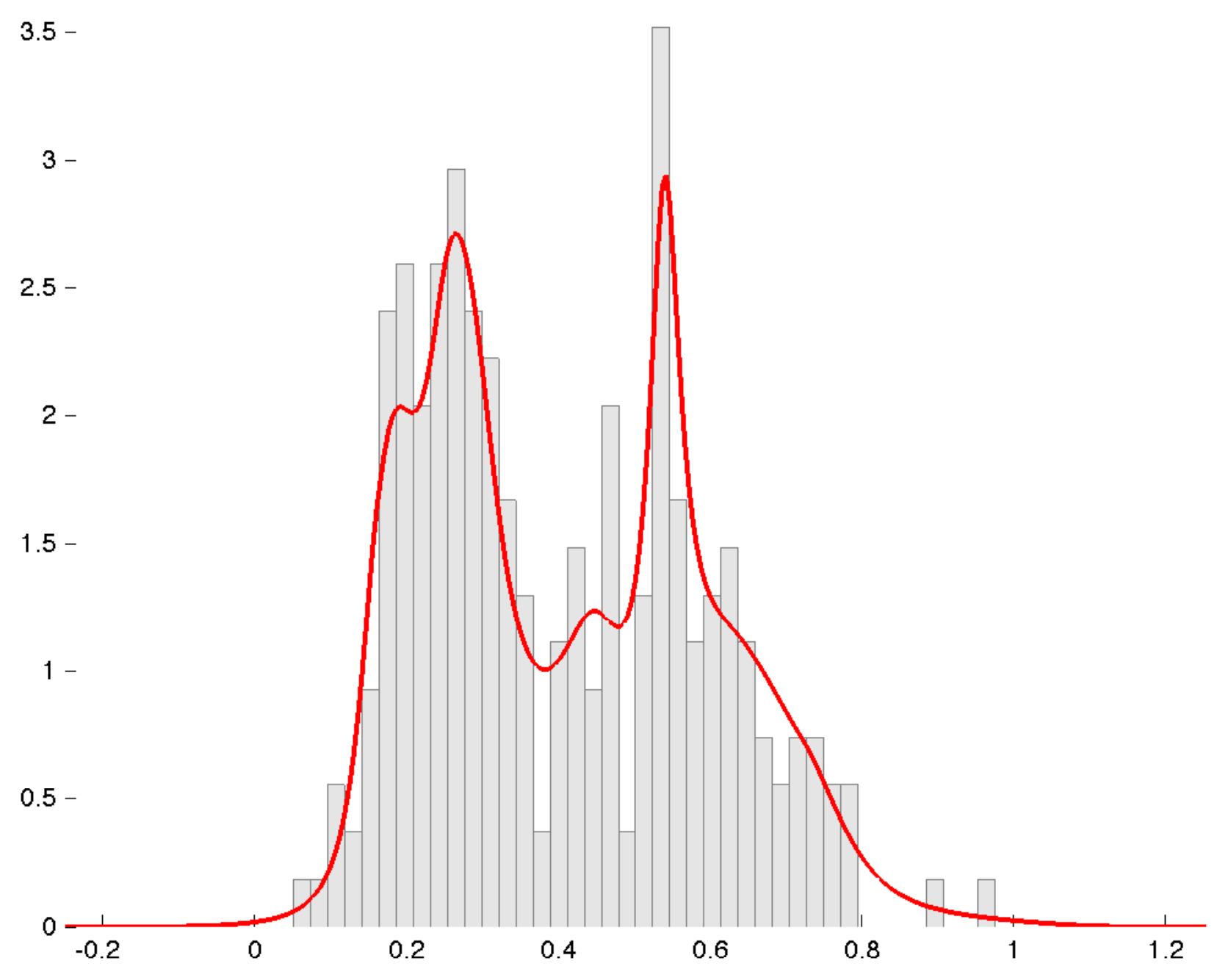}
\includegraphics[height = 2.2in]{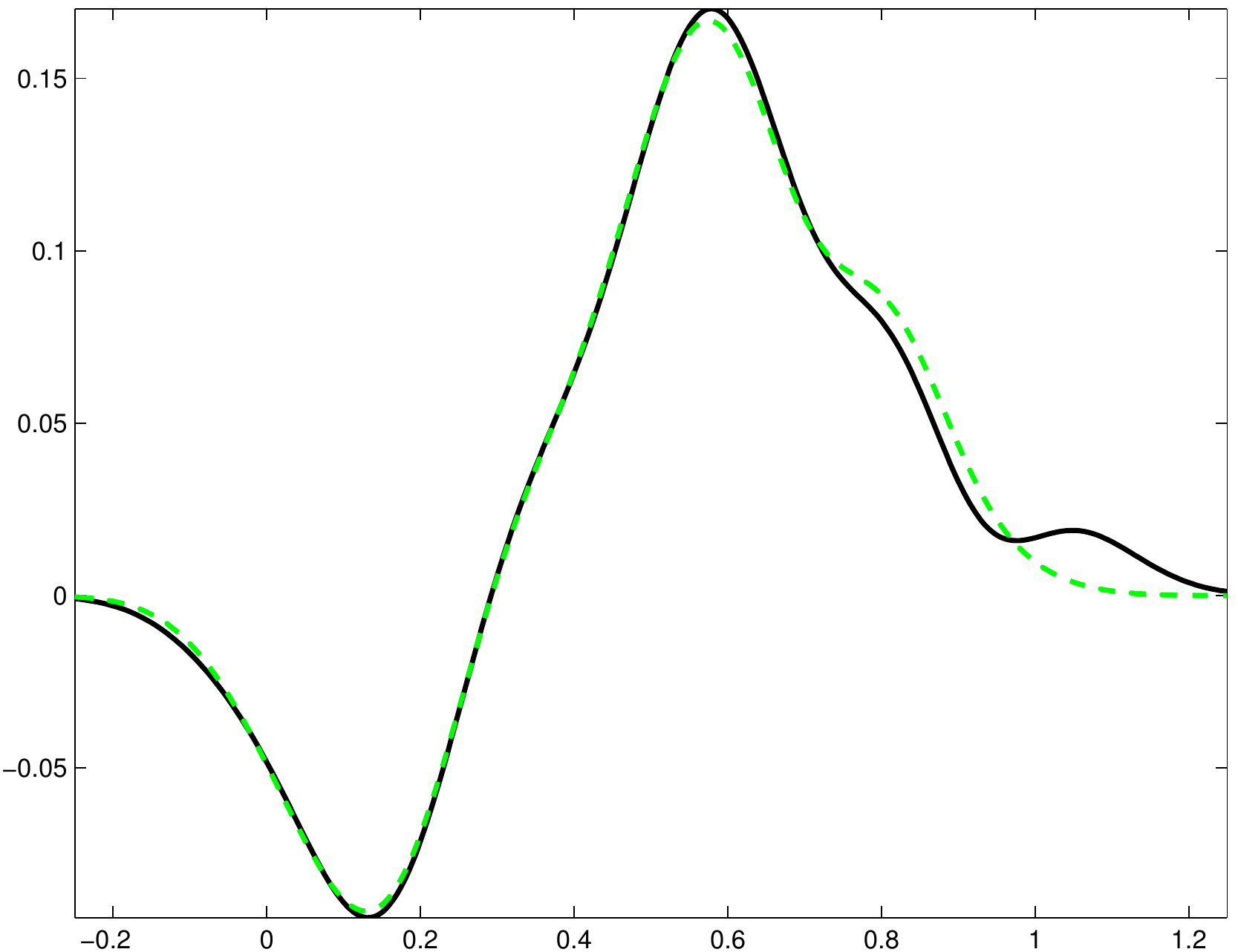}
\caption{ 
 \label{f3} 
 In this example we demonstrate that a small number of knots, in (\ref{knots}), can be enough to approximate solutions to the  Euler-Lagrange equation (\ref{ELeq}). The left column of images shows two different density estimates (red), based on the same data set  of size $n=240$ (grey histogram), using two different knot configurations (top-left uses $240$ knots, bottom-left uses $20$ knots). The right column of images show the corresponding  diagnostic curves which characterize the richness of the approximating subclass generated by the knots. The fact that the two diagnostic curves shown bottom-right are nearly identical suggests that the $20$ knots  used generate the approximating subclass by (\ref{knots}) is sufficiently rich to reach the stationary points of the penalized log-likelihood $E_\lambda$ given in (\ref{energy1}).
See Section \ref{npe} for details.
 }
 \end{figure}

To generate  knots which are sufficiently rich, in this first example,  we apply a discrete approximation at initial time $t=0$ to the gradient term $ \nabla_{y} R(x,y) $ appearing in  the Euler-Lagrange equation (\ref{ELeq}).
 For this approximation we use $N=3n$ knots in the pattern given by the following approximation
\begin{align}
 \hat v_0(x)&=  \frac{1}{\lambda n}\sum_{k=1}^n \beta_{k,0}^T R(x,X_{k})  +  \frac{1}{\lambda n}\sum_{k=1}^n   \nabla_{y} R(x,y)\Bigr|_{y= X_{k}} \\
 & \approx \frac{1}{\lambda n}\sum_{k=1}^n \beta_{k,0}^T R(x,\kappa_{k})  +  \frac{1}{\delta \lambda n}\sum_{k=1}^n   R(x,\kappa_{n+k})-R(x,\kappa_{2n+k}) \\
 & = \sum_{k=1}^N \eta_{k}^T R(x,\kappa_{k}) \label{eq:etakappa}
\end{align}
where 
$\kappa_{k}\equiv\begin{cases}
X_{k} & \text{if $k \in 1 \dots n$} \\
X_{k}+\frac{\delta}{2} & \text{if $k \in n+1 \dots 2n$} \\
X_{k}-\frac{\delta}{2} & \text{if $k \in 2n+1\dots 3n$} 
\end{cases}$ and $\eta_{k}\equiv\begin{cases}
\frac{1}{\lambda n} \beta_{k,0} & \text{if $k \in 1 \dots n$} \\
\frac{1}{\delta \lambda n} & \text{if $k \in n+1 \dots 2n$} \\
-\frac{1}{\delta \lambda n} & \text{if $k \in 2n+1 \dots 3n$} \\
\end{cases}$
with $\delta=10^{-4}$. 
 With this new set of knots, the resulting PMLE over the new class is show at bottom left in Figure~\ref{f2}. Notice that now the diagnostic function $\lambda v_0 -\mathcal D_0^v$ (the difference between the black and green line in the bottom-right plot of Figure~\ref{f2}) is much closer to zero. Indeed, for every $t\in [0,1]$ the diagnostic function $\lambda v_t -\mathcal D_t^v$ is similarly close to zero (not pictured). This implies that the maximal rate of change  of the penalized log-likelihood within the infinite dimensional Hilbert space $V^{[0,1]}$, at our estimate, is very small and hence our knots are sufficiently rich.

In the previous example we used $N=3n$ knots in (\ref{knots}) to construct a sufficiently rich class for solving the Euler-Lagrange equation (\ref{ELeq}). Now we demonstrate that with larger data sets and smaller smoothness penalties one can actually use a smaller set of knots, $N\ll n$, to approximate the solutions to Euler-Lagrange equation (\ref{ELeq}). The histograms in the left column of Figure~\ref{f3} show $n=240$ {\em iid}  samples from the same truncated mixture of normals used in the previous example. The resulting density estimates, shown in red,  use a smoothness penalty set to $\lambda=1/4$. The estimate shown top-left  utilizes  $N=n=240$ knots set at the data points whereas the estimate shown bottom-left  uses $N=20$ knots randomly selected from the data.  The right column 
shows the corresponding diagnostic plots ($\lambda v_0$ shown in black  and $\mathcal D_0^v$ shown in green). The relative agreement of the diagnostic curves in the bottom-right plot suggests that 20 knots are reasonably adequate for finding approximate solutions to the Euler-Lagrange equation. We expect this situation to improve as the number of data points increase. This has the potential to dramatically decrease the computational load when applying this estimate to extremely large data sets.

\section{Semiparametric example} 
\label{spe}

In this section we demonstrate how the PMLE $\phi_1^{\hat v}$ can be used to generate semiparametric estimation procedures obtained by assuming a parametric model on the target distribution $\Bbb P$ then introduce a nonparametric diffeomorphism to the target model. 
Indeed, any parametric model $\{\Bbb P_\theta\colon \theta\in \Theta\subset \Bbb R^m\}$ can  be extended to a   semiparametric class by considering diffeomorphisms of the data to the parametric target as follows: $\{\Bbb P_\theta \circ \phi^{v}_1 \colon \theta\in \Theta, v\in V^{[0,1]}\}$.
Since the  model $X_1,\ldots, X_n\overset{iid}\sim \Bbb P_\theta\circ \phi^v_1$ implies $\phi^v_1(X_1),\ldots, \phi_1^v(X_n)\overset{iid}\sim \Bbb P_\theta$ it is natural to alternate the optimization of $\theta$ and $\phi$ to compute the estimates $\hat \theta$ and $\hat \phi$ under this semiparametric model. This optimization routine is outlined explicitly in Algorithm \ref{alg2}.

\begin{algorithm}[h!]
\caption{Compute the semiparametric estimates $\hat\theta, \hat \phi$} 
\label{alg2} 
\begin{algorithmic}[1]
\STATE {{Set} $i=0$ and {initialize} $(\theta^0,\phi^0)$.}
\STATE { {Set}  $\phi^{i+1}$ to the PMLE of $\phi$ defined in Section \ref{pmle} under the model $X_1, \ldots, X_n \overset{iid}\sim \Bbb P_{\theta^i}\circ \phi$}.
\STATE {{Set} $\theta^{i+1}$ to the maximum likelihood estimate  of $\theta\in \Theta$ under the following model for the transformed data points:
\[ \phi^{i+1}(X_1), \ldots, \phi^{i+1}(X_n) \overset{iid}\sim \Bbb P_{\theta}\]}
\STATE {{If} $\theta^i\approx \theta^{i+1}$ and $\phi^i\approx \phi^{i+1}$ {then return} $(\hat\theta, \hat\phi)\leftarrow(\theta^{i+1},\phi^{i+1})$; {else
set} $i\leftarrow i+1$ and {return} to \mbox{step {\footnotesize 2}.}}
\end{algorithmic}
\end{algorithm}

\begin{figure}[t]
\centering
\includegraphics[height = 2.2in]{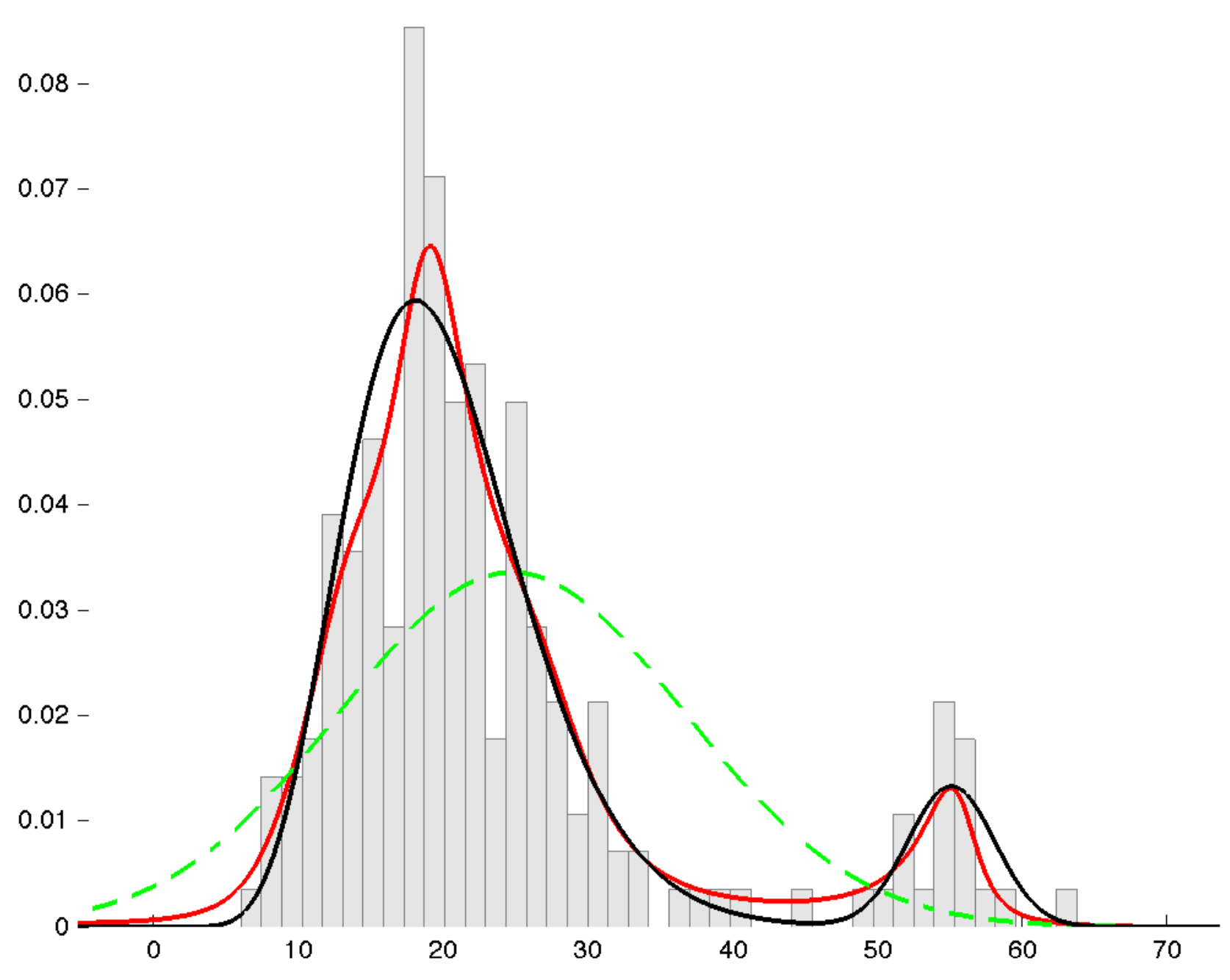}
\includegraphics[height = 2.2in]{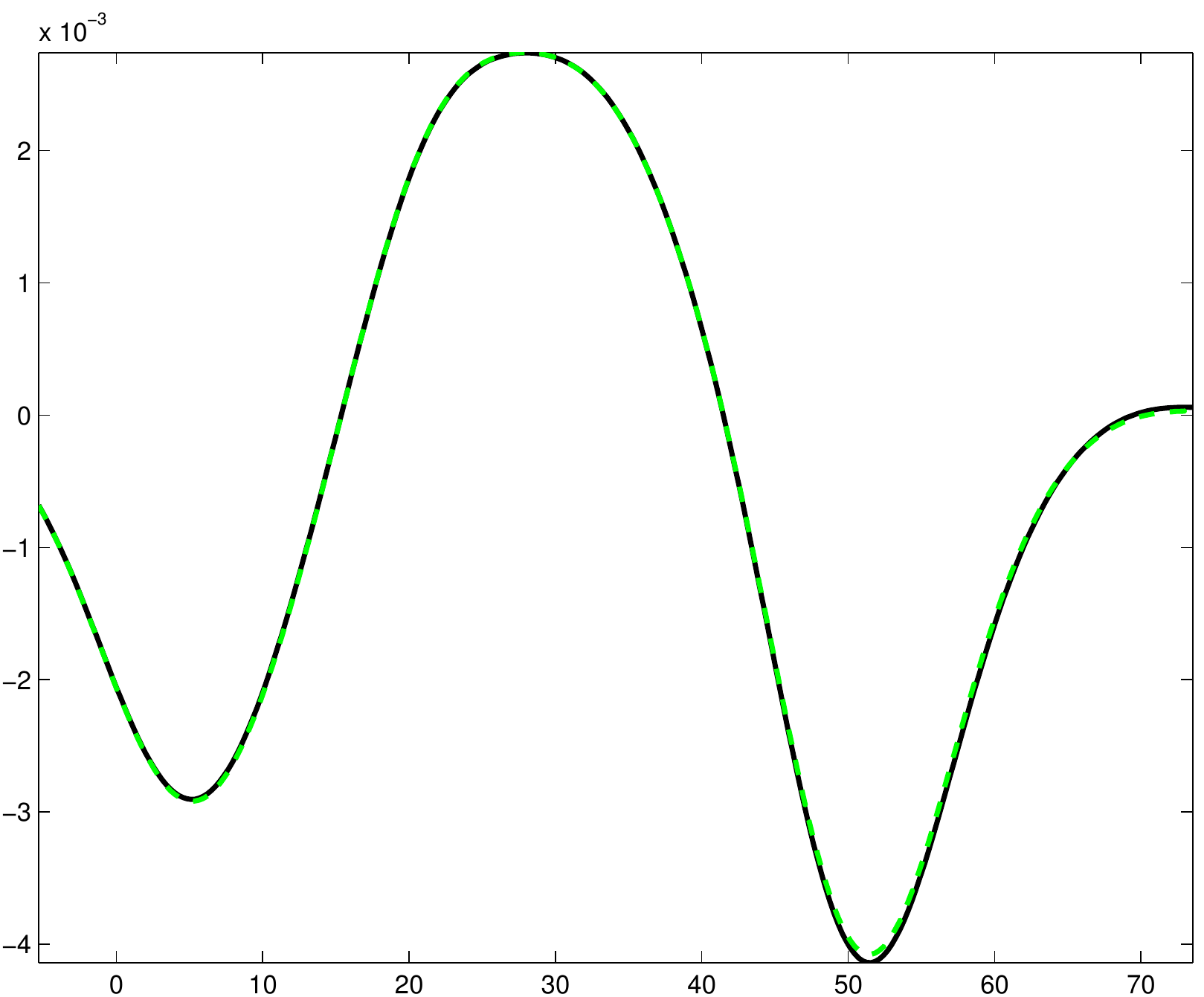}\\
\includegraphics[height = 2.2in]{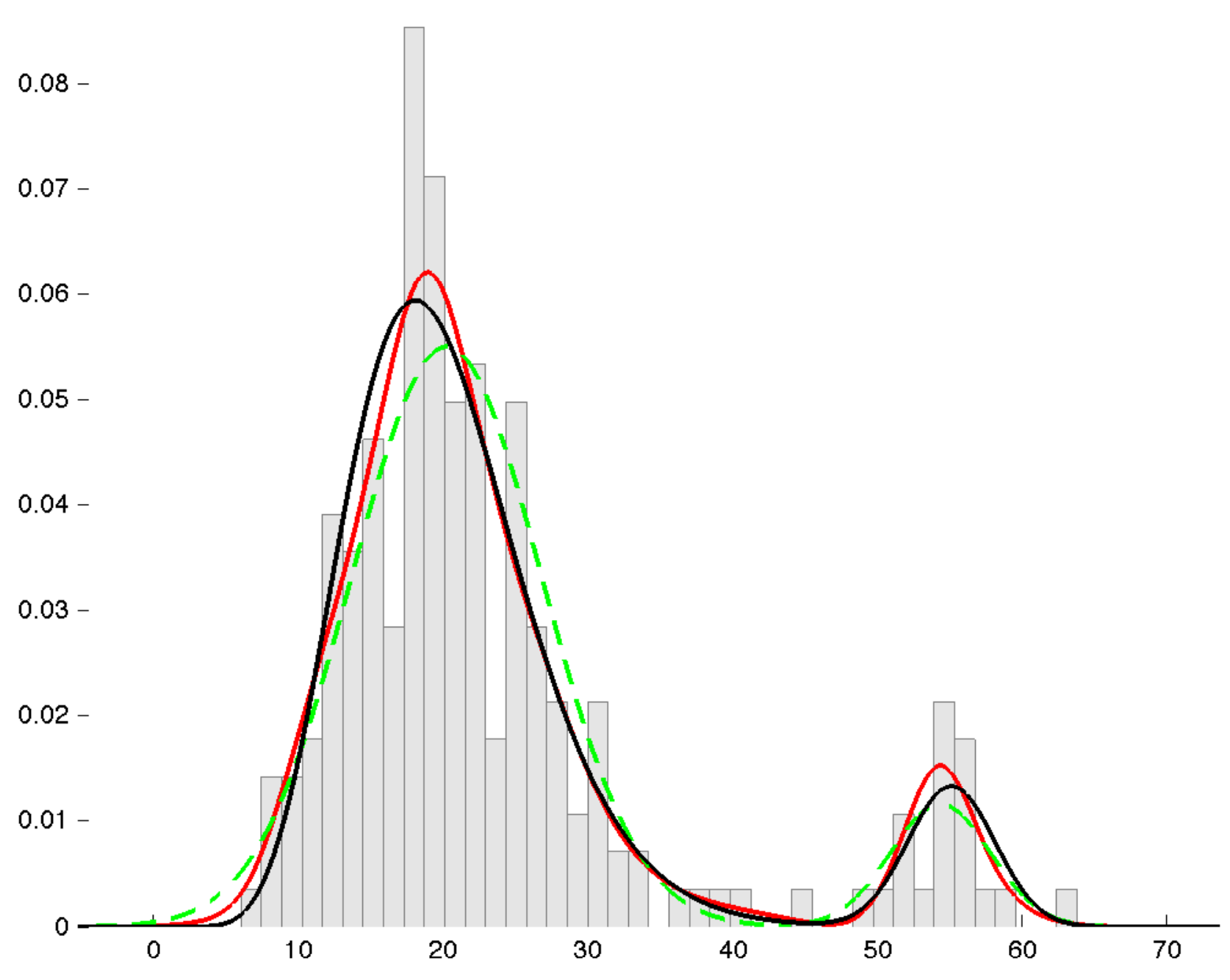}
\includegraphics[height = 2.2in]{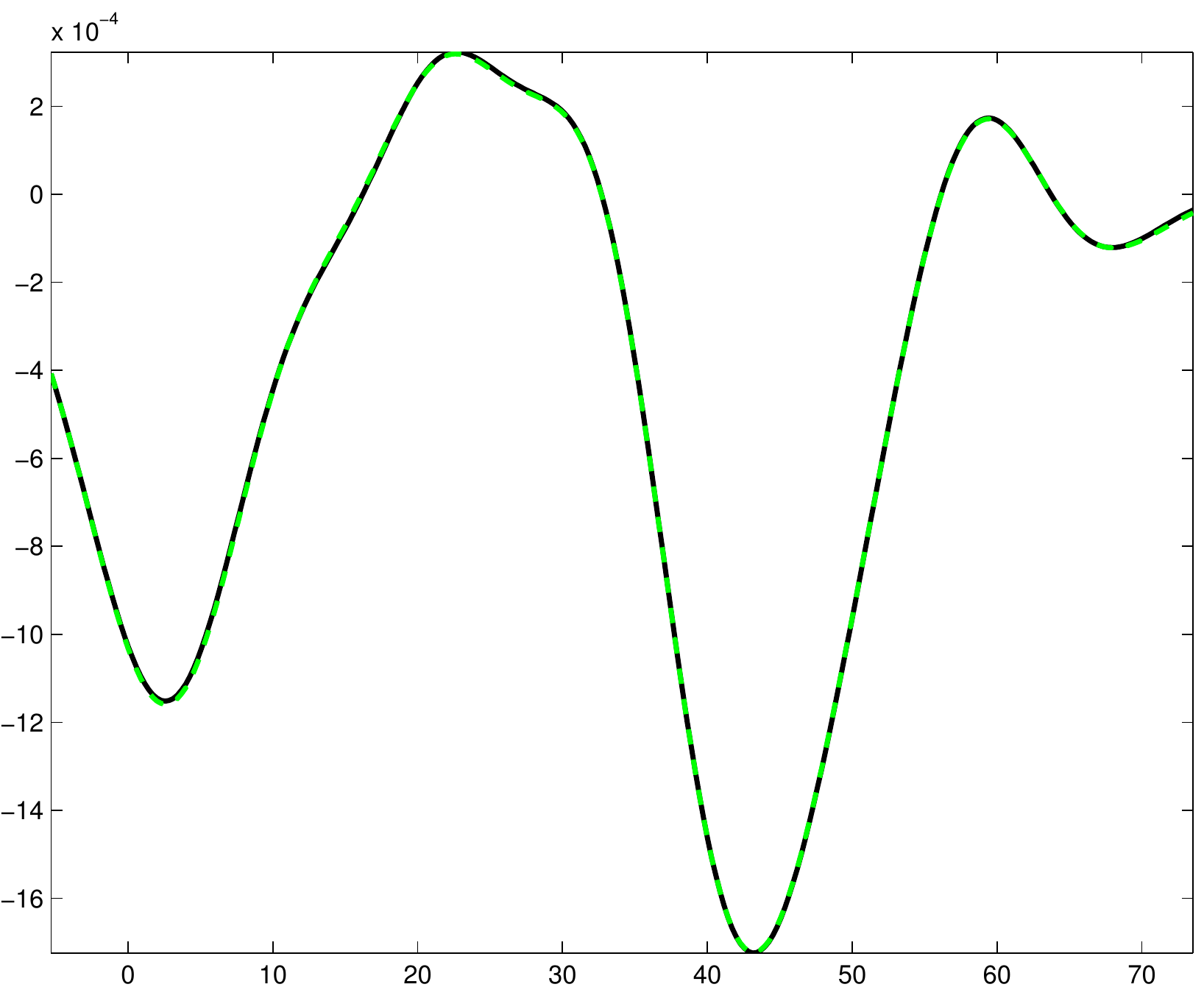}
\caption{ 
In this example we demonstrate that by parametrically modeling the target distribution one can  produce flexible semiparametric density estimates. 
The left column of histograms show the data (the same histogram plotted twice) sampled from the population density shown in black. Two semiparametric estimates are shown in red which correspond to different parametric targets. The estimated target distribution is shown in green on the left column of images. The right column of images show the corresponding  diagnostic curves which characterize the richness of the approximating subclass generated by the knots ($\lambda v_0$ is plotted in  black and $\mathcal D_0^v$ is plotted in dashed-green).  See Section \ref{spe} for details. 
 \label{f4} }
 \end{figure}

To illustrate the semiparametric nature of our estimate we sample from a population density which is a mixture of a $\chi^2$ density (with $20$ degrees of freedom) and a Gaussian density ($\mu = 55$ and $\sigma = 3$) shown in black on the left column of plots in Figure \ref{f4}. The data  comprises $n=200$ independent samples from this mixture, the histogram of which is shown on the left column of plots in Figure \ref{f4}.
 We consider two different semiparametric estimates of the population density. The first uses a basic  location-scale Gaussian family to the parametric target model
 $\{ \Bbb P_\theta \colon \theta \in\Theta \} \equiv \{ \Bbb G_{\mu,\sigma} \colon   \mu \in \Bbb R,  \sigma \in \Bbb R^+ \}$ where $\Bbb G_{\mu,\sigma}$ denotes the Gaussian measure centered at $\mu$ with variance $\sigma^2$.
 The second example uses a
 mixture of two Gaussian measures $\{ \Bbb P_\theta \colon \theta \in\Theta \} \equiv \{ \alpha\,\Bbb G_{\mu_1,\sigma_1} +(1-\alpha)\,\Bbb G_{\mu_2,\sigma_2} \colon  \mu_1, \mu_2 \in \Bbb R,  \sigma_1,\sigma_2 \in \Bbb R^+,   0< \alpha < 1 \}$.
As in Section \ref{npe} we use a Gaussian reproducing kernel to generate the Hilbert space $V$. In this example, however, we use a wider kernel, with standard deviation set to half the sample standard deviation of the data.  This is done to illustrate the flexibility in the estimated density obtained by simply changing the kernel width and the penalty parameter $\lambda$ (which is decreased to $1/2500$ in this example).
Wider kernels tend to produce estimates which have restricted local variability but can still have sufficient flexibility to model large amplitude variations over large spatial scales.

The estimate obtained from the basic location-scale Gaussian family is shown in red in the top-left plot of Figure \ref{f4}. Conversely, the estimated density which uses the mixture target model is shown in red in the bottom-left plot of Figure \ref{f4}.   The corresponding estimated target density $d\Bbb P_{\hat \theta}/dx$ is shown in green on the left two plots.
To numerically approximate the PMLE initial velocity field $\hat v_0$, needed in step 2 of Algorithm \ref{alg2},  we used the approximating subclass for the initial velocity field given in the form (\ref{knots}) with $200$ knots located at the data values.
The corresponding time zero diagnostic plots are shown in the right column of Figure \ref{f4} (green for $\lambda v_0$ and black for $\mathcal D_0^v$).  
 Notice that in both cases the semiparametric estimates do a good job at estimating the population density. 
In the case of the location-scale Gaussian target the estimated target density does a poor job of explaining the true density. However, the presence of a nonparametric diffeomorphism allows this model to fit nearly as well as a fit from a mixture model. Notice also that the semiparametric estimate based on the location-scale Gaussian target overestimates the true sampling density between the two modes. This seems due to the fact that the estimation procedure prefers an overly dispersed target density which  allows the estimated diffeomorphism to effectively add  mass around the smaller mode. The situation seem to be corrected when using a mixture.

\appendix

\section{Technical Details}
\label{TD}


This section serves to present some technical details which are used  in the proofs of Claim \ref{claim1} and Claim \ref{claim2}.   Some of these results can be found in the current literature (for example, Proposition \ref{Hspace}, most of Proposition \ref{PropDef} and equation (\ref{111}) can be found in \cite{you:10}).   However, the main goal of this section is to establish  equation (\ref{888}) in Proposition \ref{proo}  which is key to establishing Claim \ref{claim2}. 
We mention that all of the derivations presented in this section rely heavily on techniques developed by  Younes, Tr\'ouve, Miller and co-authors (see \cite{you:10} and references therein).

\if\Ver\LongVer{ 
{\flushleft\textcolor{blue}{$\downarrow$---------begin long version---------}}\newline

\begin{definition}
Here are some basic definitions and basic facts. Let $\Omega \subset \Bbb R^d$ be an open set. 
\begin{enumerate}
\item $C^0(\Omega)$ is the set of continuous functions on $\Omega$.
\item $C^0(\bar\Omega)$ is the set of continuous functions on $\bar\Omega$.
\item $C^k(\Omega)$ is the set of functions have all derivatives of order $\leq k$ continuous in $\Omega$.
\item $C^k(\bar\Omega)$ is the set of functions in $C^k(\Omega)$ all of whose derivatives of order $\leq k$ have continuous extensions to $\bar\Omega$.
\item The {\em support } of a function $f$ is the closure of the set on which $f\neq 0$.
\item A complete normed linear space is called a {\em Banach space}.
\item For $f\in C^k(\bar\Omega)$ set 
\[ \| f \|_{C^k(\bar \Omega)}\equiv \| f \|_{k,\infty}\equiv \sum_{j=0}^k \sup_{|\beta|=j}\sup_{\Omega} |D^\beta f|.\]
\item If $f\in C^\alpha(\bar\Omega)$ and $g\in C^\beta(\bar\Omega)$ then $fg \in C^{\beta\wedge \alpha}(\bar\Omega)$ and 
\[ \|fg  \|_{\beta\wedge \alpha,\infty}\leq c \|f  \|_{ \alpha,\infty}\|g  \|_{\beta,\infty}  \]
where $c$ only depends on $d,\alpha, \beta$.
\item  $C^k(\bar\Omega)$ are Banach spaces with respect to the norm $\| \cdot \|_{k,\infty}$.
\end{enumerate}
\end{definition}

{\flushleft\textcolor{blue}{$\uparrow$------------end long version---------}}\newline
} \fi

To set notation  let $C^k(\Omega, \Bbb R^d)$ denote the set of functions, mapping an open set $\Omega\subset \Bbb R^d$ into $\Bbb R^d$, which have continuous derivatives of order $\leq k$ (so that $C^0(\Omega, \Bbb R^d)$ is the continuous functions on $\Omega$ mapping into $\Bbb R^d$).
Also let $C^k(\bar\Omega,\Bbb R^d)$ denote the set of functions in $C^k(\Omega,\Bbb R^d)$ whose derivatives of order $\leq k$ have continuous extensions to $\bar\Omega$. Finally, $C_0^k(\Omega,\Bbb R^d)$ is the set of functions in $C^k(\bar\Omega,\Bbb R^d)$ whose derivatives of order $\leq k$ take the value $0$ on the boundary $\partial \Omega$. It is a well known fact that  $C^k(\bar\Omega,\Bbb R^d)$ is a Banach space with respect to the norm:  $ \| f \|_{C^k(\bar \Omega)}\equiv \| f \|_{k,\infty}\equiv \sum_{j=0}^k \sup_{|\beta|=j}\sup_{\Omega} |D^\beta f|$.

{\em Remark:} The norm $\| v\|_V\equiv \int_0^1 \| v_t\|^2_V dt $ given in Definition \ref{defV01} is technically only a semi-norm since one is free to change $v_t$  on a set of $t\in[0,1]$ with Lebesque measure zero (and not effect the norm on $V^{[0,1]}$). This is easily fixed by identifying  $V^{[0,1]}$ with the set of equivalence classes of measurable functions where $\{v_t\}_{t\in [0,1]}$ and $\{w_t\}_{t\in [0,1]}$ are said to be in the same equivalence class if $\| v_t-w_t \|_V=0$ for almost every $t\in [0,1]$. For the remainder of the paper we treat this identification as implicit with the understanding that $\{v_t\}_{t\in[0,1]}$  denotes a representer of the equivalence class to which is belongs. The following proposition establishes the Hilbert space structure of $V^{[0,1]}$ (stated without proof in 8.17 of \cite{you:10}).

\begin{proposition}
\label{Hspace}
If $V$ is a Hilbert space with inner product $\langle\cdot, \cdot \rangle_V$, then
$V^{[0,1]}$ is a Hilbert space with inner product  defined by $\langle v,h \rangle_{V^{[0,1]}} \equiv\int_0^1 \langle  v_t, h_t\rangle_V dt$.
\end{proposition}

\begin{proof}
Notice first that $\| v_t\|_V$ and $\langle  v_t, h_t\rangle_V$ are measurable functions of $t$ (this can be taken to be implicit in definitional requirement for membership in $V^{[0,1]}$: that $\int_0^1 \| v_t \|^2_Vdt <\infty$). Now, with the exception of completeness, all the properties of a Hilbert space inner product are inherited from $\langle  \cdot , \cdot\rangle_V$ and the linear properties of Lebseque integration over $[0,1]$. 
To show completeness  let $v^n\in V^{[0,1]}$ be a Cauchy sequence so that $\int_0^1 \| v_t^m - v_t^n \|^2_V dt\rightarrow 0$. By the completeness of $V$ there exists a Borel set $B\subset[0,1]$ such that for all $t\in B$ there exists a $v_t\in V$ such that $\|v_t^n - v_t\|_V\rightarrow 0$. On $t\in [0,1]\setminus B$ we are free to set $v_t\equiv 0$ (the zero element of $V$). For this $v_t$ we have that $\|v^n_t - v_t \|_{V^{[0,1]}}^2\equiv \int_0^1 \| v_t^n - v_t \|^2_V dt\rightarrow 0$. Therefore $V^{[0,1]}$ is complete.
\end{proof}

\begin{proposition} 
\label{PropDef} Let $\Omega$ be an open bounded subset of $\Bbb R^d$ and $V$ be a Hilbert space such that  $V\hookrightarrow C_0^1(\Omega,\Bbb R^d)$. If   $v\in V^{[0,1]}$ then there exists a unique  class of $C^1$ diffeomorphisms of $\Omega$,  $\{\phi_{t}^v\}_{t\in [0,1]}$, such that  $\phi_t^v(x)\in C^0 ([0,1]\times \overline \Omega,\Bbb R^d)$ and  which satisfy the ordinary differential equation $
 \partial_t \phi_t^v(x) = v_t(\phi_t^v(x)) $ with boundary condition $\phi_0^v(x)=x$, for all $x\in \Omega$. 
Moreover,
 \begin{equation}
 \label{div}
  \log\det D\phi_{st}^v(x)  = \int_s^t  \text{\rm div}\,  v_u (\phi_{su}^v(x))du. 
  \end{equation}
 \end{proposition}
\begin{proof} 
%
\if\Ver\LongVer{ 
{\flushleft\textcolor{blue}{$\downarrow$---------begin long version---------}}\newline

We first show that there exists a $\delta>0$ and a solution $\phi^v_t(x)$   to the ordinary differential equation $ \partial_t \phi_t^v(x) = v_t(\phi_t^v(x)) $ on $(t,x)\in [0,\delta]\times\Omega$  with boundary condition $\phi_0^v(x)=x$ (for all $x\in \Omega$) which is in $C^0 ([0,\delta]\times \overline \Omega,\Bbb R^d)$.  For any map $\varphi_t(x)\in C^0 ([0,\delta]\times \overline \Omega,\Bbb R^d)$ we define the operator $\Gamma(\varphi)\equiv x + \int_0^t v_s(\varphi_s(x))ds $ (extending $v_s$ to $0$ beyond $\Omega$ if necessary). Notice  $\Gamma$ maps $C^0 ([0,\delta]\times \overline \Omega,\Bbb R^d)$ into $C^0 ([0,\delta]\times \overline \Omega,\Bbb R^d)$  since 
\begin{align*}
|x + \int_0^t v_s(\varphi_s(x))ds  - y - \int_0^w v_s(\varphi_s(y))ds |
&\leq |x-y| + \int_t^w | v_s(\phi_s(x))-v_s(\varphi_s(y))  |ds \\
&\leq |x-y| + c\int_t^w \| v_s \|_V|\phi_s(x)-\varphi_s(y)  |ds \\
&\rightarrow 0, \quad\text{by DCT as $t\rightarrow w$ and $x\rightarrow y$}
\end{align*}
and also by the fact that $|x + \int_0^t v_s(\varphi_s(x))ds| \leq \sup_{x\in \Omega}|x| +  c\int_0^1 \|v_s\|_V ds <\infty$.
We show that $\delta$ can be chosen to make $\Gamma$ a contraction with respect to sup norm on $C^0 ([0,\delta]\times \overline \Omega,\Bbb R^d)$
\begin{align}
\|\Gamma(\varphi) - \Gamma(\varphi^\prime) \|_{\infty} &\leq \int_0^t |v_s(\varphi_s(x))-  v_s(\varphi^\prime_s(x)) |ds \nonumber\\
&\leq \int_0^t c\|v_s\|_V|\varphi_s(x)-  \varphi^\prime_s(x) |ds\nonumber \\
&\leq \|\varphi-  \varphi^\prime \|_\infty  \int_0^\delta c\|v_s\|_V ds \label{Contract}
\end{align}
If we choose $\delta$ so that $\int_0^\delta c\|v_s\|_V ds < \gamma<1$ then the above inequality shows that $\Gamma$ is contraction operator on the Banach space $C^0 ([0,\delta]\times \overline \Omega,\Bbb R^d)$. In fact, for reasons that will become obvious later we notice that since $\int_0^{t} c\| v_s \|_Vds$ is uniformly continuous for all $t\in[0,1]$ we can choose a $\delta$ such that $c\int_{t}^{t+\delta} \|v_s\|_V ds < \gamma<1$ uniformly over all $t\in [0,1]$. Therefore there exists a unique fixed point $\phi_t(x) \in C^0 ([0,\delta]\times \overline \Omega,\Bbb R^d)$ which satisfies
\[ \phi_t(x) = x + \int_0^t v_s(\phi_s(x))ds. \]
Now to extend  the definition of $\phi_t(x)$ to $ t\in [0,2\delta]$ notice one can repeat the same argument to get the existence of $\phi_{\delta t}(x)$ which satisfies $\partial_t \phi_{\delta t}(x) = v_t(\phi_{\delta t}(x)) $ on $(t,x)\in [\delta,2\delta]\times\Omega$  with boundary condition $\phi_{\delta \delta}(x)=\phi_\delta(x)$.
In particular simply set 
\begin{equation}
\phi_t^v(x)\equiv \begin{cases}
\phi_t(x) & \text{when $t\in[0,\delta]$}\\
\phi_{\delta t}(x) & \text{when $t\in[\delta,2\delta]$}
\end{cases}
\end{equation}
Now repeat the argument over successive intervals $[k\delta, (k+1) \delta]$. Notice that this can lead to kinks at the endpoints $k\delta$, but this can be fixed by working with overlapping intervals and using the uniqueness of the fixed point.
This method then produces a unique   $\phi_t^v(x)\in C^0 ([0,1]\times \overline \Omega,\Bbb R^d)$ and  which satisfy the ordinary differential equation $
 \partial_t \phi_t^v(x) = v_t(\phi_t^v(x)) $ with boundary condition $\phi_0^v(x)=x$, for all $x\in \Omega$.

Now we need to show that for all $t\in [0,1]$, $\phi_t$ is a $C^1$ diffeomorphism of $\Omega$. 

{\flushleft\textcolor{blue}{$\uparrow$------------end long version---------}}\newline
} \fi

First note that if $v\in V^{[0,1]}$ and $V\hookrightarrow C_0^1(\Omega, \Bbb R^d)$ then $ \| v_t \|_{1,\infty} \leq c \| v_t \|_{V} $. Now by H\"older, $\int_0^1 \| v_t \|_{V} dt \leq \|v \|_{V^{[0,1]}}<\infty$  so that the arguments for  Theorem 8.7   in \cite{you:10} to apply to the class  $ V^{[0,1]}$. In particular, there exists a unique  class of $C^1$ diffeomorphisms of $\Omega$, $\phi_t^v(x)\in C^0 ([0,1]\times \overline \Omega,\Bbb R^d)$, which satisfy the ordinary differential equation $\partial_t \phi_t^v(x) = v_t(\phi_t^v(x))$
 with boundary condition $\phi_0^v(x)=x$, for all $x\in \Omega$. Moreover, by Proposition 8.8 in \cite{you:10} we have that
 \begin{align}
    \partial_t D\phi_{st}^v(x)  =  D v_t (\phi_{st}^v(x))  D\phi_{st}^v(x) \label{gg}
\end{align}
where $\det D\phi_{ss}^v(x)=Id_d$. 
Since $D\phi_{st}^v(x)$ is nonsingular and differentiable in $t$ we have that (see  (6.5.53) of \cite{hor:91}, for example)
\begin{align*}
\partial_t \log \det D\phi_{st}^v(x)&= \text{trace}\bigl\{ [D\phi_{st}^v(x)]^{-1}   \partial_t D\phi_{st}^v(x)  \bigr\} \\
& = \text{trace}\bigl\{ [D\phi_{st}^v(x)]^{-1}  D v_t (\phi_{st}^v(x))  D\phi_{st}^v(x) \bigr\} ,\,\,\text{by (\ref{gg})} \\
& = \text{trace}\bigl\{ D v_t (\phi_{st}^v(x))  \bigr\} \\
& = \text{div}\, v_t (\phi_{st}^v(x)).
\end{align*}
Therefore $\log \det D\phi_{st}^v(x)$ is differentiable everywhere on $t\in [0,1]$ with derivative given by $ \text{div}\, v_t (\phi_{st}^v(x))$.  

Since $v_t(x)$ is measurable with respect to both arguments $t$ and $x$ (by definition) and limits of measurable functions are measurable, the function $\text{div}\,v_t(x)$  is also measurable. Since $\phi^v_{st}(x)$ is continuous with respect to both $t$ and $x$,  $\text{div}\, v_t (\phi_{st}^v(x))$ is also measurable. 
Notice that $\text{div}\, v_t (\phi_{st}^v(x))$  is also Lebesque integrable since $|\text{div}\, v_t (\phi_{st}^v(x))| \leq c\| v_t \|_V$ by the embedding $V\hookrightarrow C_0^1(\Omega,\Bbb R^d)$ and the fact that that  $\int_0^1 \| v_t \|_{V} dt \leq \|v \|_{V^{[0,1]}}<\infty$.
Therefore by Theorem 7.21 of \cite{rud:66} we have that
\[\log \det D\phi_{st}^v(x) =  \int_s^t  \text{div}\, v_u (\phi_{su}^v(x))  du\]
since $\log \det D\phi_{ss}^v(x) = 0$. \end{proof}

\begin{lemma} If $V\hookrightarrow C_0^1(\Omega,\Bbb R^d)$ and $v,w\in V^{[0,1]}$, then
\begin{align}
 \|\phi_{st}^v - \phi_{st}^w    \|_{\infty}& \leq c \| v-w \|_\text{\tiny $V^{[0,1]}$} \exp{\left( c\| v \|_\text{\tiny $V^{[0,1]}$} \right)}
  \label{222}
\end{align}
where $c$ is a constant which does not depend on $v, w, s$ or $t$.
Moreover, if we additionally suppose $V\hookrightarrow C_0^2(\Omega, \Bbb R^d)$ then
\begin{align}
  \label{777}
\| \phi_{st}^v - \phi_{st}^w \|_{1,\infty}\leq \| v-w \|_\text{\tiny $V^{[0,1]}$}  F\bigl( {\|v\|_\text{\tiny $V^{[0,1]}$}},{\|w\|_\text{\tiny $V^{[0,1]}$}}\bigr)
\end{align}
where $F(\cdot,\cdot)$ is  a finite  function on $\Bbb R\times \Bbb R$, monotonically increasing in both arguments, which does not depend on $v, w, s$ or $t$. 
\end{lemma}

\begin{proof}
The inequality (\ref{222})  follows directly from Grownwell's lemma  applied to the following inequality
\begin{align*}
|\phi_{st}^v(x) - \phi_{st}^w(x)| &= \left| \int_{s}^t v_u (\phi_{su}^v(x)) -  w_u (\phi_{su}^w(x)) du\right|\\
&\leq \int_{s}^t |v_u (\phi_{su}^v(x)) -  v_u (\phi_{su}^w(x)) |du + \int_{s}^t |v_u (\phi_{su}^w(x)) -  w_u (\phi_{su}^w(x)) |du\\
& \leq \int_{s}^t c\|v_u\|_{V} |\phi_{su}^v(x) -  \phi_{su}^w(x)|du  +  c\|v - w  \|_{V^{[0,1]}}
\end{align*}
where the last inequality follows from the assumption $V\hookrightarrow C_0^1(\Omega, \Bbb R^d)$.
\if\Ver\LongVer{ 
{\flushleft\textcolor{blue}{$\downarrow$---------begin long version---------}}\newline
 Therefore 
 \[  |\phi_{st}^v(x) - \phi_{st}^w(x)| \leq  c_2\|v - w  \|_{V^{[0,1]}} \exp\left( \int_{s}^t \|v_u\|_{1,\infty}  du \right) \]
 by Grownwell's lemma.  This proves (\ref{222}).
 {\flushleft\textcolor{blue}{$\uparrow$------------end long version---------}}\newline
} \fi

 To prove (\ref{777}) notice that for any vector $h\in \Bbb R^d$ we have that $\partial_t D \phi_{st}^v(x)h= Dv_t (\phi_{st}^v(x))D \phi_{st}^v(x)h $ where $ D \phi_{ss}^v(x)h = h$ (by Proposition 8.8 in \cite{you:10} and also  \cite{dup:98}). Therefore
 \begin{align}
 \label{uugg}
 D \phi_{st}^v(x)h -  D \phi_{st}^w(x)h = \int_s^t \Bigl[  Dv_u (\phi_{su}^v(x))D \phi_{su}^v(x)h -  Dw_u (\phi_{su}^w(x))D \phi_{su}^w(x)h \Bigr]\,du
 \end{align}
 where we are using the fact that $D \phi_{st}^v(x)h$  is differentiable with respect to $t$ everywhere in $[0,1]$ and with  Lebseque integrable derivative   (and using Theorem 8.21 of \cite{rud:66}). 
 Now notice that the integrand of (\ref{uugg}) satisfies
  \begin{align*}
\bigl| &Dv_u (\phi_{su}^v(x))D \phi_{su}^v(x)h -  Dw_u (\phi_{su}^w(x))D \phi_{su}^w(x)h  \bigr| \leq I + I\!I
\end{align*}
where 
\begin{align*}
I &=  \left| Dv_u (\phi_{su}^v(x)) \Bigl\{ D \phi_{su}^v(x)h - D \phi_{su}^w(x)h \Bigr\}  \right| \leq   c\bigl\| v_u\bigr\|_{V} \Bigl|   D \phi_{su}^v(x)h - D \phi_{su}^w(x)h  \Bigr| 
\end{align*}
and
\begin{align*}
I\!I & =  \left|\Bigl\{ Dv_u (\phi_{su}^v(x)) -  Dw_u (\phi_{su}^w(x))  \Bigr\} D \phi_{su}^w(x)h  \right| \\
&\leq   \Bigl\{ \| v_u \|_{2,\infty} \| \phi_{su}^v - \phi_{su}^w \|_\infty   + \bigl\| v_u-w_u\bigr \|_{1,\infty} \Bigr\} \,  \bigl\| \phi_{su}^w\bigr\|_{1,\infty} \, |h |\\
&\leq \Bigl\{   c \| v_u \|_{V^{\phantom{[}}}  \|v-w\|_{V^{[0,1]}} \exp\left(c\| w \|_{V^{[0,1]}} \right)   + c\bigl\| v_u-w_u\bigr \|_{V} \Bigr\}\,  \bigl\| \phi_{su}^w\bigr\|_{1,\infty} \, |h |
 \end{align*}
\if\Ver\LongVer{ 
{\flushleft\textcolor{blue}{$\downarrow$---------begin long version---------}}\newline
We are using the fact that for any $v_t\in V$ we have that $\bigl |Dv_t (x) h \bigr |\leq \| v_t \|_{1,\infty } | h| $ for any $ h,  x \in \Bbb R^d$.  Also we are using the fact that for any $v_t\in V$ we have that $\bigl |(Dv_t (x) - Dv_t(y))h \bigr |\leq \| v_t \|_{2,\infty }|x-y| | h| $ for any $ h,  x \in \Bbb R^d$.  
  {\flushleft\textcolor{blue}{$\uparrow$------------end long version---------}}\newline
} \fi
where the last inequality follows by (\ref{222}). To bound $I\!I$ further notice $
  \|\phi_{st}^v\|_{1,\infty} \leq c_1\exp{\left( c\| v \|_\text{\tiny $V^{[0,1]}$} \right)}$. To see why,
apply Gronwall's lemma to the following inequality
\begin{align*}
|\phi_{st}^v(x) - \phi_{st}^v(y)| &= \left| x-y + \int_{s}^t v_u (\phi_{su}^v(x)) -  v_u (\phi_{su}^v(y)) du\right|\\
& \leq  |x-y| + \int_{s}^t c\, \|v_u\|_{V^{\phantom{[}}} |\phi_{su}^v(x) -  \phi_{su}^v(y)| du
\end{align*}
which yields $| \phi_{st}^v(x) - \phi_{st}^v(y) |\leq |x-y| \exp{\left( c\| v \|_\text{\tiny $V^{[0,1]}$} \right)}$.
  Since $\phi_{st}^v(x)$ is differentiable with respect to $x$ everywhere in $\Omega$, for each multi-index $\beta$ such that $|\beta|=1$ there exists a direction $h^\beta\in \Bbb R^d$ (with $|h^\beta|=1$) such that 
 \begin{align*}
 |D^\beta \phi_{st}^v(x) | &= \lim _{\epsilon \downarrow 0} \frac{|\phi_{st}^v(x+\epsilon h^\beta) - \phi_{st}^v(x)|}{\epsilon}  \leq  \exp{\left( c\| v \|_\text{\tiny $V^{[0,1]}$} \right)}.
 \end{align*} 
 Combining the above inequality with the fact that $ \|\phi_{st}^v\|_\infty \leq \sup_{x\in \Omega} |x| $ gives  the desired inequality 
 $ \|\phi_{st}^v\|_{1,\infty} \leq c_1\exp{\left( c\,\| v \|_\text{\tiny $V^{[0,1]}$} \right)}$. Applying this to $ I\!I $ gives
 \begin{align*}
 I\!I &\leq \Bigl\{  c \| v_u \|_{V}  \|v-w\|_{V^{[0,1]}} \exp\left(c \| w \|_{V^{[0,1]}} \right)   + c \bigl\| v_u-w_u\bigr \|_{V} \Bigr\}\, c_1\exp{\left( c \| w \|_\text{\tiny $V^{[0,1]}$} \right)} |h|
 \end{align*}
 Therefore
 \begin{align*}
\int_s^t I\!I\, du 
&\leq c_1 c |h| \|v-w \|_{V^{[0,1]}}\Bigl\{\|v\|_{V^{[0,1]}} \exp\bigl( 2c\| w \|_{V^{[0,1]}} \bigr)  +  \exp\bigl( c\| w \|_{V^{[0,1]}} \bigr)  \Bigr\}\\
& =  |h|  \|v-w\|_{V^{[0,1]}}  F\left(\|v \|_{V^{[0,1]}} ,\| w\|_{V^{[0,1]}}\right) 
\end{align*}
where $F(x,y)$ is monotone and finite in both $x$ and $y$. Now by equation (\ref{uugg}) we have that 
\begin{align*} 
\bigl|D \phi_{st}^v(x)h -  D \phi_{st}^w(x)h\bigr | & \leq \int_s^t I du + \int_s^t I\!I du \\
&\leq  \int_s^t c\bigl\| v_u\bigl\|_{V} \bigl|   D \phi_{su}^v(x)h - D \phi_{su}^w(x)h  \bigr| du \\
&\qquad\qquad +  |h|  \|v-w\|_{V^{[0,1]}}  F\bigl(\|v \|_{V^{[0,1]}},\| w\|_{V^{[0,1]}}\bigr).
\end{align*}
By Gronwell's lemma we have that
\[\bigl|D \phi_{st}^v(x)h -  D \phi_{st}^w(x)h\bigr | \leq  |h|  \|v-w\|_{V^{[0,1]}}  F\bigl(\|v \|_{V^{[0,1]}},\| w\|_{V^{[0,1]}}\bigr) \exp\bigl ( c \bigl\| v\bigl\|_{V^{[0,1]}} \bigr) . \]
Now by taking a supremum over $x\in \Omega$, $|h|=1$ and combining with (\ref{222}) gives (\ref{777}), after redefining $F$ to accommodate the extra term $\exp ( c \bigl\| v\bigl\|_{V^{[0,1]}} )$.


\end{proof}

\begin{proposition} 
\label{proo}
If $v,h \in V^{[0,1]}$ and $V\hookrightarrow C_0^1(\Omega,\Bbb R^d)$ then for all $x\in \Omega$ and $s,t\in [0,1]$ 
\begin{align}
{\partial_\epsilon}  \phi^{ v+\epsilon h}_{st}(x)  &= \int_s^t  \bigl\{D\phi^{ v+\epsilon h}_{ut} h_u \bigr\}\circ{\phi^{ v+\epsilon h}_{su}(x)}\,   du. \label{111}
 \end{align}
\if\Ver\LongVer{ 
{\flushleft\textcolor{blue}{$\downarrow$---------begin long version---------}}\newline
If, in addition, $V\hookrightarrow C_0^2(\Omega,\Bbb R^d)$ then ${\partial_\epsilon}  \phi^{ v+\epsilon h}_t(x)  $ is locally Lipschitz continuous in $\epsilon$   (with respect to sup-norm over $x\in \Omega$). 
   {\flushleft\textcolor{blue}{$\uparrow$------------end long version---------}}\newline
} \fi
If, in addition,  $V\hookrightarrow C_0^3(\Omega,\Bbb R^d)$ then
\begin{align}
  \label{888}
\partial_\epsilon  \log \det D\phi_{1}^{v+\epsilon h }(x)\bigr|_{\epsilon = 0} & = \int_0^1  \Bigl[ h_u \cdot \nabla \log\det D\phi_{u1}^v  + \text{\rm div}\, h_u\Bigr] \circ  \phi^v_{u}(x)  \, du
\end{align}
\end{proposition}

\begin{proof} The assumption that $v,h \in V^{[0,1]}$ and $V\hookrightarrow C_0^1(\Omega,\Bbb R^d)$ are sufficient to apply
Theorem 8.10 of \cite{you:10} which gives
\begin{align*}
{\partial_\epsilon}  \phi^{ v+\epsilon h}_{st}(x) \bigr|_{\epsilon = 0} &= \int_s^t  \bigl\{D\phi^{ v}_{ut} h_u \bigr\}\circ{\phi^{ v}_{su}(x)}\,   du. 
 \end{align*}
 Now since ${\partial_\epsilon}  \phi^{ v+\epsilon h}_{st}(x) = {\partial_\xi}  \phi^{ v+\epsilon h+\xi h}_{st}(x)\bigr|_{\xi = 0} $ one immediately obtains (\ref{111}). 
 
\if\Ver\LongVer{ 
{\flushleft\textcolor{blue}{$\downarrow$---------begin long version---------}}\newline

    Now we can prove the local Lipschitz continuity of  ${\partial_\epsilon}  \phi^{ v+\epsilon h}_t  $ in $\epsilon$  (with respect to sup-norm). Let $\xi, \epsilon\in (-M, M)$ for some $0<M<\infty$ and fix $v,h\in V^{[0,1]}$ and write
\begin{align}
\left \| \partial_\epsilon  \phi^{ v+\epsilon h}_{st} -  \partial_\xi   \phi^{ v+\xi h}_{st}  \right\|_{\infty} & = \left \| \int_s^t  \bigl\{D\phi^{ v+\epsilon h}_{ut} h_u \bigr\}\circ{\phi^{v+\epsilon h}_{su}(x)}-  \bigl\{D\phi^{v+\xi h}_{ut} h_u \bigr\}\circ{\phi^{v+\xi h}_{su}(x)}\,   du    \right\|_{\infty}\nonumber \\
& \leq  \left\| \int_s^t  \bigl\{D\phi^{ v+\epsilon h}_{ut} h_u \bigr\}\circ{\phi^{ v+\epsilon h}_{su}(x)}-  \bigl\{D\phi^{ v+\xi h }_{ut} h_u \bigr\}\circ{\phi^{v+\epsilon h}_{su}(x)}\,   du    \right\|_{\infty}\label{hhhh1}  \\ 
& \qquad+ \left\| \int_s^t  \bigl\{D\phi^{v+\xi h}_{ut} h_u \bigr\}\circ{\phi^{v+\epsilon h}_{su}(x)}-  \bigl\{D\phi^{v+\xi h}_{ut} h_u \bigr\}\circ{\phi^{ v+\xi h}_{su}(x)}\,   du    \right\|_{\infty}  \label{hhhh2}
\end{align}
The integrand of the first  term (\ref{hhhh1}) can be bounded by noticing that   $V\hookrightarrow C_0^{2}(\Omega, \Bbb R^d)$ implies
\begin{align}
  \Bigl| \bigl\{D\phi^{ v+\epsilon h}_{ut} h_u \bigr\}\circ{\phi^{ v+\epsilon h}_{su}(x)}- & \bigl\{D\phi^{ v+\xi h }_{ut} h_u \bigr\}\circ{\phi^{v+\epsilon h}_{su}(x)}\Bigr| \nonumber\\
  &\leq \|  (D\phi_{ut}^{ v+\xi h} -D \phi_{ut}^{ v+\epsilon h})h_u   \|_\infty  \nonumber\\
  &\leq  \|  \phi_{ut}^{ v+\xi h} - \phi_{ut}^{v+\epsilon h}  \|_{1,\infty}  \| h_u\|_{\infty} \nonumber\\
&\leq |\xi -\epsilon| \|  h  \|_{V^{[0,1]}}  F\bigl( {\|v+\xi h\|_\text{\tiny $V^{[0,1]}$}},{\|v+\epsilon h\|_\text{\tiny $V^{[0,1]}$}}\bigr)   \| h_u\|_{V},\,\,\text{by (\ref{777})} \nonumber\\
&\leq c_1 |\xi -\epsilon|   \| h_u\|_{V}
\label{uuuu}
\end{align}
where $c_1$ is a finite constant which depends on $v$ and $h$ but not on   $\xi$ or $\epsilon$.
The integrand of the second term  (\ref{hhhh2}) can be bounded as follows:
\begin{align}
 \Bigr|\bigl\{D\phi^{v+\xi h}_{ut} h_u \bigr\}\circ{\phi^{v+\epsilon h}_{su}(x)}-  &\bigl\{D\phi^{v+\xi h}_{ut} h_u \bigr\}\circ{\phi^{ v+\xi h}_{su}(x)}\Bigl| \nonumber   \\
  &\leq   \bigl\| D\phi^{v+\xi h}_{ut} h_u \bigr\|_{1,\infty} \bigl| \phi^{v+\epsilon h}_{su}(x)-  \phi^{ v+\xi h}_{su}(x)  \bigr| \nonumber   \\
 &\leq  c_2 |\xi -\epsilon|  \| D\phi^{v+\xi h}_{ut} h_u \|_{1,\infty} ,\,\,\text{by (\ref{222})} \nonumber   \\
&\leq  c_2 |\xi -\epsilon|   \| h_u \|_{1,\infty} \| \phi^{v+\xi h}_{ut} \|_{2,\infty}  \nonumber   \\
&\leq  c_2 c_3|\xi -\epsilon|   \| h_u \|_{V}   \label{fine} 
\end{align}
where $c_2, c_3$  are finite constants which depends on $v$ and $h$ but not on   $\xi$ or $\epsilon$ (the existence of $c_3$ follows again from  Theorem 8.9 of \cite{you:10}). Combining (\ref{uuuu}) and (\ref{fine}) with (\ref{hhhh1}) and (\ref{hhhh2}) shows that $ \partial_\epsilon  \phi^{ v+\epsilon h}_{st}$  is locally Lipschitz in $\epsilon$.

   {\flushleft\textcolor{blue}{$\uparrow$------------end long version---------}}\newline
} \fi

 To show (\ref{888}) notice  that partial derivatives on $x$ can pass under the integral in (\ref{111})  to compute $D {\partial_\epsilon}  \phi^{ v+\epsilon h}_{1}$. This follows by first noticing that  $V\hookrightarrow C_0^{2}(\Omega, \Bbb R^d)$ implies
  \begin{align}
  \sup_{u\in[0,1]} \bigl\| \phi_{ut}^{v+\epsilon h} \bigr\|_{2,\infty}&\leq c_1 \exp\left({c_2 \|  {v\|_{V^{[0,1]}}+M\| h} \|_{V^{[0,1]}}}\right) \label{811}
  \end{align}
  for all $|\epsilon|< M$, by equation (8.11) of \cite{you:10}.
\if\Ver\LongVer{ 
{\flushleft\textcolor{blue}{$\downarrow$---------begin long version---------}}\newline
Here is the above inequality with an extra step.
 \begin{align}
  \sup_{u\in[0,1]} \| \phi_{ut}^{v+\epsilon h} \|_{p+1,\infty}&\leq c_1 \exp\left({c_2 \|  {v+\epsilon h} \|_{V^{[0,1]}}}\right)\\
  &\leq c_1 \exp\left({c_2 \|  {v\|_{V^{[0,1]}}+M\| h} \|_{V^{[0,1]}}}\right) <\infty \label{811}
  \end{align}
 Note that   (8.11) of  \cite{you:10} can be proved using
 \begin{align}
 \label{rrmark}
  \| \phi_{st}^v(x) - \phi_{ss}^v(x) \|_{p,\infty}
  &\leq\int_s^t \|  v_u(\phi_{su}^v(x))   \|_{p,\infty} du
 \end{align}
 and then applying Proposition 8.4 of \cite{you:10} along with Gronwall's identity and induction. Notice that when $p\geq 1$ equation (\ref{rrmark}) can be obtained, not from the identity $\phi_{st}^v(x) - \phi_{ss}^v(x)= \int_s^t  v_u(\phi_{su}^v(x))    du$, but from the identity
 \[ D^\beta\phi_{st}^v(x) - D^\beta\phi_{ss}^v(x)= \int_s^t    D^\beta(v_u(\phi_{su}^v(x)))   du.\]
 which follows from Proposition 8.8 of  \cite{you:10}.
   {\flushleft\textcolor{blue}{$\uparrow$------------end long version---------}}\newline
} \fi
Therefore when fixing $v, h\in V^{[0,1]}$ the function $ \| D\phi^{v+\epsilon h}_{ut} h_u \|_{1,\infty}$ is bounded above by a finite constant over $(u,\epsilon)\in[0,1]\times (-M,M)$. With an additional application of Proposition 8.4 of \cite{you:10} we also have that
  $ \bigl\| \{D\phi^{v+\epsilon h}_{ut} h_u\}\circ \phi_{su}^{v+\epsilon h} \bigr\|_{1,\infty}<\infty$ uniformly over $(u,\epsilon)\in[0,1]\times (-M,M)$.
Therefore, indeed, partial derivatives on $x$ can pass under the integral in (\ref{111})  to obtain
\begin{equation}
\label{pass}
D {\partial_\epsilon}  \phi^{ v+\epsilon h}_{1}(x) = \int_0^1 D\bigl[ \{D\phi_{u1}^{v+\epsilon h} h_u  \}\circ \phi_u^{v+\epsilon h} (x)\bigr]  du. 
\end{equation}

Now we show that $D {\partial_\epsilon}  \phi^{ v+\epsilon h}_{1}(x)$ is continuous  over $(x,\epsilon)\in \Omega\times (-M,M)$. This will allow us to switch the order of $D$ and $\partial_\epsilon$ and establish (\ref{888}). The same reasoning which allows $D$ to pass under the integral in (\ref{111})  also allows us to pass limits on $x$ and $\epsilon$ under the integral in (\ref{pass}). Therefore it will be sufficient to show the integrand, $ D\bigl[ \{D\phi_{u1}^{v+\epsilon h} h_u  \}\circ \phi_u^{v+\epsilon h}(x)  \bigr]$,  in (\ref{pass}) is continuous  over $(x,\epsilon)\in \Omega\times (-M,M)$.
To see why   the integrand in (\ref{pass})  is continuous first note that $ \phi^{ v+\epsilon h}_{st}$ is a $C^2$ diffeomorphism (by a similar proof Theorem 8.7 in \cite{you:10}). Secondly, under the assumption $V\hookrightarrow C_0^3(\Omega, \Bbb R^d)$ one can extend (\ref{777}) to bound  $\| \phi_{st}^{v+\epsilon h} - \phi_{st}^{v+\xi h} \|_{2,\infty}$ by $c|\xi - \epsilon|$, where $c$ is a finite constant which may depend on $v,h$ but not on $\xi, \epsilon$. These two facts imply that $ D\bigl[ \{D\phi_{u1}^{v+\epsilon h} h_u  \}\circ \phi_u^{v+\epsilon h}(x)  \bigr]$ is indeed continuous in $(x,\epsilon)\in \Omega\times (-M,M)$ which implies that $D {\partial_\epsilon}  \phi^{ v+\epsilon h}_{1}(x)$ is also. 

The continuity of $D {\partial_\epsilon}  \phi^{ v+\epsilon h}_{1}(x)$ over $(x,\epsilon)\in \Omega\times (-M,M)$ implies  $ \partial_\epsilon D \phi^{ v+\epsilon h}_{st}$ exists and   $D {\partial_\epsilon}  \phi^{ v+\epsilon h}_{st}  =  {\partial_\epsilon} D \phi^{ v+\epsilon h}_{st}$ (see \cite{cou:36}, page 56). Then since $D\phi_{st}^{v+\epsilon h}$ is nonsingular (by the diffeomorphic property) and differentiable  with respect to $\epsilon$ we have that
\begin{align*}
\partial_\epsilon  \log \det D\phi_{st}^{v+\epsilon h }(x)&= \text{trace}\bigl\{ [D\phi_{st}^{v+\epsilon h}(x)]^{-1}   \partial_\epsilon D\phi_{st}^{v+\epsilon h}(x) \bigr\} \\
 &= \text{trace}\bigl\{ [D\phi_{st}^{v+\epsilon h}(x)]^{-1}  D \partial_\epsilon \phi_{st}^{v+\epsilon h}(x)  \bigr\} .
\end{align*}
Therefore, by (\ref{pass}),
\begin{align*}
\partial_\epsilon  \log \det D\phi_{1}^{v+\epsilon h }(x)\bigr|_{\epsilon = 0} 
& =\text{trace}\left\{  \bigl[ D\phi_1^v(x) \bigr]^{-1} \int_0^1 D\bigl[ \{D\phi_{u1}^v h_u  \}\circ \phi_{1u}^v\circ \phi^v_1(x)  \bigr]  du  \right\} \\
& =\text{trace}\left\{  \bigl[ D\phi_1^v(x) \bigr]^{-1} \int_0^1 D\bigl[ \{D\phi_{u1}^v h_u  \}\circ \phi_{1u}^v(y)\bigr]\Bigr|_{y=\phi^v_1(x)}     du\, D\phi^v_1(x)  \right\} \\
& = \int_0^1\text{trace}\left\{  D\bigl[ \{D\phi_{u1}^v h_u  \}\circ \phi_{1u}^v(y)\bigr]\Bigr|_{y=\phi^v_1(x)}   \right\}    du .  
\end{align*}
Now notice that
\begin{align*}
\text{trace}\, D \bigl[ \{  D\phi_{u1}^v h_u\}\circ \phi^v_{1u}\bigr] 
&= \text{trace}\, \bigl[ \{  D (D\phi_{u1}^v h_u)\}\circ \phi^v_{1u} D(\phi^v_{1u})\bigr] \\
&= \text{trace}\, \bigl[    \{D  (D\phi_{u1}^v  h_u) \} (D\phi_{u1}^v )^{-1}\bigr]\circ \phi^v_{1u} \\
 &=\bigl\langle  h_u\circ  \phi^v_{1u} ,(\nabla \log\det D\phi_{u1}^v )\circ  \phi^v_{1u} \bigr\rangle_d + (\text{div}\, h_u) \circ  \phi^v_{1u} .
\end{align*}
The last line follows from the identity: $\text{trace}\, \bigl[    \{D  [D\phi_{u1}^v h_u] \} (D\phi_{u1}^v)^{-1}\bigr]= \bigl\langle  h_u,\nabla \log\det D\phi_{u1}^v \bigr\rangle_d + \text{trace}(D h_u)$. 
\if\Ver\LongVer{ 
{\flushleft\textcolor{blue}{$\downarrow$---------begin long version---------}}\newline
$\text{trace}\, \bigl[    \{D  D\phi h \} (D\phi)^{-1}\bigr]= \bigl\langle  h,\nabla \log\det D\phi \bigr\rangle_d + \text{div}\, h$. This is easily seen to be true using a symbolic mathematical program. For example the following {\sc{Matlab}} code works
\begin{quote}
\tt{syms x y;\\
f1=sym('f1(x,y)');\\
f2=sym('f2(x,y)'); \\
h1=sym('h1(x,y)'); \\
h2=sym('h2(x,y)'); \\
F=[f1;f2]; \\
h=[h1;h2]; \\
DF=jacobian(F,[x y]);\\
LHS=trace( jacobian(DF*h,[x y])*inv(DF) );\\
RHS=(jacobian(log(det(DF)),[x y])  )*h + trace(jacobian(h,[x y]));\\
simplify(LHS-RHS)\\
\%Note that this simplifies to zero!!!!!}
\end{quote}
{\flushleft\textcolor{blue}{$\uparrow$------------end long version---------}}\newline
} \fi
Therefore
 \begin{align*}
\partial_\epsilon  \log \det D\phi_{1}^{v+\epsilon h }(x)\bigr|_{\epsilon = 0} & = \int_0^1  \Bigl[ h_u \cdot \nabla \log\det D\phi_{u1}^v   + \text{div}\, h_u  \Bigr]\circ  \phi^v_{u}(x)  du.
\end{align*}

\end{proof}
\if\Ver\LongVer{ 
{\flushleft\textcolor{blue}{$\downarrow$---------begin long version---------}}\newline
\section{Radial kernel derivatives}
The following lemma shows how to simplify the evolution of $q,m, A$ and $b$ in the above equations when the reproducing kernel $R(x,y)$ is a radial kernel $R(|x-y|)$.
\begin{lemma}
Let $R$ be a real valued function defined on $[0,\infty)$ such that  $R(|\cdot|)\in C^4(\Bbb R)$. If $x\neq y$ are in $\Bbb R^d$, then
\begin{align}
\nabla_x R(|x-y|)&=R^\prime(|x-y|)\frac{(x-y)}{|x-y|} \label{yy1}\\
\nabla_y \otimes \nabla_x R(|x-y|)&=(x-y)\otimes (y-x) \left[ \frac{R^{\prime\prime}(|x-y|)}{|x-y|^2}  -  \frac{R^{\prime}(|x-y|)}{|x-y|^3} \right]  - \frac{R^{\prime}(|x-y|) }{|x-y|} \text{\rm Id}_d  \label{yy2} \\
\nabla_{y} [\nabla_y \cdot \nabla_x R(|x-y|)]&= \frac{y-x}{|x-y|}\left[-R^{\prime\prime\prime}(|x-y|)+ (1-d)\frac{R^{\prime\prime}(|x-y|)}{|x-y|} -   (1-d)\frac{R^{\prime}(|x-y|)}{|x-y|^2}  \right].  \label{yy3}
\end{align}
When $x=y$,
\begin{align}
\nabla_x R(|x-y|)\Bigr|_{x=y}&=0   \label{yy4}\\
\nabla_y \otimes \nabla_x R(|x-y|)\Bigr|_{x=y}&= -\text{\rm Id}_d  R^{\prime\prime}(0)  \label{yy5}\\
\nabla_{y} [\nabla_y \cdot \nabla_x R(|x-y|)]\Bigr|_{x=y}&=0.    \label{yy6}
\end{align}
Finally notice that $\nabla_y \otimes \nabla_x R(|x-y|)=-\nabla_x \otimes \nabla_x R(|x-y|)$.
\end{lemma}

\begin{proof}
The equation (\ref{yy1}) holds by noticing that $\nabla_x |x-y|= (x-y)/|x-y|$. Then equation  (\ref{yy4}) holds since $R^{\prime}(|x-y|)\rightarrow 0$ as $|x-y|\rightarrow 0$. The second equation (\ref{yy2}) follows since
\[ \nabla_y \otimes \nabla_x R(|x-y|)=  \nabla_y  [(x-y) F(|x-y|)]=-e_i F(|x-y|)+ F^\prime(|x-y|) \frac{(y-x)\otimes (x-y)}{|x-y|}\]
where $F(y)=R^\prime(y)/y$ and $e_i=(0,\ldots,1,\ldots,0)$ (with the $1$ in the $i^\text{th}$ coordinate) so that
\[ F^\prime(y)=\frac{R^{\prime\prime}(y)y- R^{\prime}(y)}{y^2} =\frac{R^{\prime\prime}(y)}{y}-\frac{ R^{\prime}(y)}{y^2}.\]
Also notice that $ F^\prime(|x-y|) \frac{(y-x)\otimes (x-y)}{|x-y|}\rightarrow 0$ as $|x-y|\rightarrow 0$ since each coordinate is bounded by $F^\prime(|x-y|) |x-y|= R^{\prime\prime}(|x-y|)- R^\prime(|x-y|)/|x-y|\rightarrow R^{\prime\prime}(0)- R^{\prime\prime}(0)=0$. Now since $F(|x-y|)\rightarrow R^{\prime\prime}(0)$ as $|x-y|\rightarrow 0$, we also get equation (\ref{yy5}).

Finally notice that (\ref{yy3}) follows since
\[ \nabla_y \cdot \nabla_x R(|x-y|)=\text{trace}\bigl\{\nabla_y \otimes \nabla_x R(|x-y|) \bigr\}=-R^{\prime\prime}(|x-y|)+(1-d)\frac{R^\prime(|x-y|)}{|x-y|}\]
Now when $x\neq y$
\[ \nabla_y [ \nabla_y \cdot \nabla_x R(|x-y|)]=-R^{\prime\prime\prime}(|x-y|)\frac{y-x}{|x-y|}+(1-d) F^\prime(|x-y|) \frac{y-x}{|x-y|}. \]
To study the case when $x=y$ we can form the difference quotient to compute  $\frac{\partial }{\partial y_i} \nabla_y \cdot \nabla_x R(|x-y|)$
\begin{align}
\left[\frac{\partial }{\partial y_i} \nabla_y \cdot \nabla_x R(|x-y|)\right]_{x-y=0}&=\frac{ \nabla_y \cdot \nabla_x R(|\epsilon e_i|)-\nabla_y \cdot \nabla_x R(0) }{\epsilon} \\
&=\frac{-R^{\prime\prime}(\epsilon)+(1-d){R^\prime(\epsilon)}/{\epsilon}  +d R^{\prime\prime}(0) }{\epsilon} \\
&=\frac{-R^{\prime\prime}(0)+ O(\epsilon^2)+(1-d){R^{\prime\prime}(0)} + O(\epsilon^2)  +d R^{\prime\prime}(0) }{\epsilon}\label{IamSmall} \\
&=O(\epsilon)
\end{align}
where (\ref{IamSmall}) following since $R^\prime(0)=R^{\prime\prime\prime}(0)=0$.
\end{proof}
{\flushleft\textcolor{blue}{$\uparrow$------------end long version---------}}\newline
} \fi

\if\Ver\LongVer{ 
{\flushleft\textcolor{blue}{$\downarrow$---------begin long version---------}}\newline

\section{Numerical techniques: geodesic and transpose flows}
Here we explicitly show how to generate Geodesic flows.

\subsection{Geodesic flow  when  $v^t(x)=\sum_{j=1}^N \eta^t_j R(x,\kappa^t_j) $ on $\Bbb R^1$}
Suppose $v^t(x)$ is a vector field on $\Bbb R$ such that $v^t(x)=\sum_{j=1}^N \eta^t_j R(x,\kappa^t_j)$.
We give formal arguments that  show how to generate the full path of coefficients $\{\eta^t_j: t\in [0,1], j=1,\ldots, N\}$ and knots $\{\kappa^t_j: t\in [0,1], j=1,\ldots, N\}$ from the $t=0$ values when the vector field flow minimizes $\frac{1}{2}\int_0^1 \| v^t  \|^2_R dt$ under the endpoint constraints $\kappa_j^{0}=X_j$ and $\kappa_j^{1}=\phi^1(X_k)$. {\em Imagine restricting all the flows to having this form, then minimizing the PMLE. The minimum fixes $\phi_\text{optim}^1$ and the flow vector field also minimizes  $\frac{1}{2}\int_0^1 \| v^t  \|^2_R dt$ under the endpoint constraints $\kappa_j^{0}=X_j$ and $\kappa_j^{1}=\phi_\text{optim}^1(X_k)$}

Start by noticing that $\frac{1}{2}\int_0^1 \| v^t  \|^2_R dt = \frac{1}{2}\int_0^1 \eta^T R \eta dt $ where $\eta\equiv (\eta_1,\ldots,\eta_N)^T$ and $R=(R(\kappa_i,\kappa_j))_{i,j = 1}^N$ (where we are suppressing the time variable $t$ dependence on $\eta_j$ and $x_j$). Notice $\dot \kappa = R \eta$ where $\kappa = (\kappa_1,\ldots, \kappa_N)^T$ so that 
\[ \frac{1}{2}\int_0^1 \| v^t  \|^2_V dt =  \frac{1}{2}\int_0^1 \dot\kappa^T R^{-1} \dot\kappa\, dt = \int_0^1 L(\kappa,\dot\kappa) dt   \]
where $L(\kappa,\dot\kappa) =\dot\kappa^T R^{-1} \dot\kappa/2$.
 Now take a time varying perturbation $\kappa^t +\epsilon h^t$ which does not perturb the endpoints. In particular $h^0(\kappa^0_k)=h^1(\kappa^1_k)=0$ for all $k=1,\ldots, N$. Now formally since vector field flow minimizes $\frac{1}{2}\int_0^1 \| v^t  \|^2_R dt$ under the endpoint constraints $\kappa_j^{0}=X_j$ and $\kappa_j^{1}=\phi_1(X_k)$ we get 
\begin{align*}
0&=\frac{d}{d\epsilon}\left[\int_0^1 L(\kappa+\epsilon h,\dot \kappa+\epsilon \dot h)dt\right]_{\epsilon = 0} 
=\int_0^1 \Bigl[ L^{(1,0)}(\kappa,\dot \kappa)\cdot h + L^{(0,1)}(\kappa,\dot \kappa)\cdot \dot h\Bigr]  dt  \\
&=\int_0^1 \left[L^{(1,0)}(\kappa,\dot \kappa) - \frac{d}{dt}L^{(1,0)}(\kappa,\dot \kappa)\right] \cdot h  \, dt 
\end{align*}
where the last line follows by integration by parts and the assumption that $h$ is zero at the endpoints (we define  $L^{(0,1)}(x,y)$ to be the row  vector with $i^\text{th}$ element $\frac{\partial L}{\partial y_i}$ and  $L^{(1,0)}(x,y)$ to be the row  vector with $i^\text{th}$ element $\frac{\partial L}{\partial x_i}$). Since the perturbation direction $h$ was arbitrary (with the endpoint constraints) we get Euler-Lagrange equation
\begin{equation}
\label{EL}
 \frac{d}{dt}L^{(1,0)}(\kappa,\dot \kappa) = L^{(1,0)}(\kappa,\dot \kappa)
\end{equation}
Now, notice  $ L^{(0,1)}(\kappa,\dot \kappa) = R^{-1}\dot \kappa = \eta$ and the $i^\text{th}$ value of $L^{(0,1)}(\kappa,\dot \kappa) $ can be computed as follows
\begin{align*}
L^{(1,0)}(\kappa,\dot \kappa)_i & = \frac{\partial}{ \partial \kappa_i} \frac{ \dot\kappa^T R^{-1} \dot\kappa}{2}  = -\frac{1}{2} \dot\kappa^T R^{-1} \frac{\partial R}{\partial \kappa_i}  R^{-1} \dot\kappa  =   -\frac{1}{2} \eta^T \frac{\partial R}{\partial \kappa_i}  \eta \\
&= -\frac{\eta_i}{2}\Bigl( \sum_{j=1}^N \eta_j R^{(1,0)}(\kappa_i,\kappa_j) + \sum_{j=1}^N \eta_j R^{(0,1)}(\kappa_j,\kappa_r)   \Bigr) \\
&= -{\eta_i} \sum_{j=1}^N \eta_j R^{(1,0)}(\kappa_i,\kappa_j) 
\end{align*}
where the last line follows when $R(x,y) = R(|x-y|)$ so that $R^{(1,0)}(x,y)= - R^{(0,1)}(x,y)$, $R^{(1,0)}(x,y)= - R^{(1,0)}(y,x)$ and $R^{(1,0)}(x,y)=  R^{(0,1)}(y,x)$.
Therefore we can simplify the Euler-Lagrange equation (\ref{EL}) to the following update equation for $\eta$
\begin{equation}
\label{updateEta} 
\frac{d\eta_i}{dt} = -{\eta_i} \sum_{j=1}^N \eta_j R^{(1,0)}(\kappa_i,\kappa_j). \end{equation}
We also trivially have the update equation for the knots
\begin{equation}
\label{updateKnots}
\frac{d\kappa_i}{dt}=\sum_{j=1}^N \eta_j R(\kappa_i,\kappa_j)  .
\end{equation}

\subsection{Perturbing $\eta$ and $\kappa$ at $t=0$, in $\Bbb R^1$}
Notice that since the full path of $\eta$ and $\kappa$ is determined at time $t=0$ one can perturb these values at time zero  $\eta^0+\epsilon \delta \eta^0$ and $\kappa^0+ \epsilon \delta \kappa^0$. If $\epsilon$ is infinitesimal, this results in perturbations $\delta \kappa$, $\delta \eta$ at all times. 
Taking the derivatives (with respect to $\epsilon$) on both sides of (\ref{updateEta}) and (\ref{updateKnots}) one gets the following linear ODE characterization of $\delta\eta$
\begin{align}
\frac{d\delta \eta_i}{dt} &= -\sum_{j=1}^N \delta\eta_i \,  \eta_j R^{(1,0)}(\kappa_i,\kappa_j) +  {\delta\eta_j}\,  \eta_i R^{(1,0)}(\kappa_i,\kappa_j) \nonumber \\
&\qquad - \sum_{j=1}^N \delta\kappa_i\, {\eta_i}  \eta_j R^{(2,0)}(\kappa_i,\kappa_j)+ \delta\kappa_j\, {\eta_i}  \eta_j R^{(1,1)}(\kappa_i,\kappa_j).
\end{align} 
Similarly one gets the following linear ODE characterization of $\delta\kappa$
\begin{equation}
\frac{d\delta\kappa_i}{dt}=\sum_{j=1}^N \delta\eta_j\, R(\kappa_i,\kappa_j) +   \delta \kappa_i\,\eta_j R^{(1,0)}(\kappa_i,\kappa_j) +  \delta \kappa_j\,\eta_j R^{(0,1)}(\kappa_i,\kappa_j).
\end{equation}
At the map level one also gets perturbations $\delta \phi$ and $\delta D\phi$ with initial conditions $\delta \phi^0 \equiv 0$ and  $\delta D \phi^0 \equiv 0$
\begin{align}
\frac{d}{dt}\phi(x) &= \sum_{j=1}^N \eta_j R(\phi(x),\kappa_j) \\
\frac{d}{dt}D\phi(x) &= D\phi(x)\sum_{j=1}^N \eta_j R^{(1,0)}(\phi(x),\kappa_j) \\
\frac{d}{dt}\delta\phi(x) &= \sum_{j=1}^N \delta\eta_j R(\phi(x),\kappa_j) + \delta \phi (x)\, \eta_j R^{(1,0)}(\phi(x),\kappa_j) +  \delta \kappa_j\, \eta_j R^{(0,1)}(\phi(x),\kappa_j)\\
\frac{d}{dt}\delta D\phi(x) &= \sum_{j=1}^N  \delta D\phi(x)\,  \eta_j R^{(1,0)}(\phi(x),\kappa_j) +  \delta  \eta_j \, D\phi(x) R^{(1,0)}(\phi(x),\kappa_j) \\
& \qquad +\delta \phi(x)\, D\phi(x) \eta_j R^{(2,0)}(\phi(x),\kappa_j)  +\delta \kappa_j \, D\phi(x) \eta_j R^{(1,1)}(\phi(x),\kappa_j)
\end{align}

Notice that these are linear equation so that we can {\em stack} the perturbation variables 
\begin{equation}
\frac{d}{dt} \left[\begin{array}{c} \delta\eta\\ \delta \kappa \\ \delta\phi \\ \delta D\phi   \end{array}\right]
= U  \left[\begin{array}{c} \delta\eta\\ \delta \kappa \\ \delta\phi \\ \delta D\phi   \end{array}\right]
\end{equation}
where $U$ is a matrix which depends on $\eta, \kappa, \phi(x) $ and $D\phi(x)$ but not on  $\delta\eta, \delta\kappa, \delta\phi(x), \delta D\phi(x)$.

To see the effect of the rate of change on the energy notice that when perturbing $\eta$ and $\kappa$ one gets a perturbation  $\phi^t(X_i) +\epsilon \delta \phi^t(X_i)$. Notice that $\delta \phi$ is a Eulerian specification of a flow field.  To switch to Lagrangian specification on defines $u^t(x)$ as satisfying the identity
\[ u^t (\phi^t(x))=\delta \phi^t(x)\]
Earlier we derive the rate of change of the log likelihood at time $t=1$ with respect to Lagrangian coordinates as follows
\[  \frac{1}{n} \sum_{i=1}^n \text{div}\, u (\phi^1(X_i)) + \langle \nabla H (\phi^1(X_i)), u^1(\phi^1(X_i))\rangle.  \]
 To switch to the Eulerian specification notice that $D\delta \phi^t (x) =  Du^t (\phi^t(x)) D\phi^t(x) $. Therefore $\text{trace} ([D\delta \phi^t (x)][ D\phi^t(x)]^{-1} ) = \text{trace}\bigl( Du^t (\phi^t(x))\bigr)  =  \text{div}\, u (\phi^1(x))$ and hence the rate of change of the log likelihood can be computed as
\[ \frac{1}{n} \sum_{i=1}^n \text{trace} ([ D\phi^t(X_i)]^{-1} [D\delta \phi^t (X_i)]) + \langle \nabla H (\phi^1(X_i)), \delta \phi^1(X_i)\rangle  \]

Also notice that the rate of the hamiltonian ....

\subsection{Transpose flow when  $v^t(x)=\sum_{j=1}^N \eta^t_j R(x,\kappa^t_j) $ on $\Bbb R^1$}

{\flushleft\textcolor{blue}{$\uparrow$------------end long version---------}}\newline
} \fi

\bibliography{refs}

\begin{thebibliography}{10}

\bibitem{alla:07}
S.~Allassonni{\`e}re, Y.~Amit, and A.~Trouv{\'e}.
\newblock Towards a coherent statistical framework for dense deformable
  template estimation.
\newblock {\em Journal of the Royal Statistical Society: Series B (Statistical
  Methodology)}, 69(1):3--29, 2007.

\bibitem{and:11}
E.~Anderes and M.~Coram.
\newblock Two-dimensional density estimation using smooth invertible
  transformations.
\newblock {\em Journal of Statistical Planning and Inference}, 141(3):1183 --
  1193, 2011.

\bibitem{aro:50}
N.~Aronszajn.
\newblock Theory of reproducing kernels.
\newblock {\em Trans. Amer. Math. Soc.}, 68:337--404, 1950.

\bibitem{Beg:2006ly}
M.~Beg and A.~Khan.
\newblock Computing anb average anatomical atlas using lddmm and geodesic
  shooting.
\newblock {\em Medical Image Analysis}, pages 1116--1119, 2006.

\bibitem{conf/miccai/BegMTY03}
M.~Beg, M.~Miller, A.~Trouv{\'e}, and L.~Younes.
\newblock The euler-lagrange equation for interpolating sequence of landmark
  datasets.
\newblock In Randy~E. Ellis and Terry~M. Peters, editors, {\em MICCAI (2)},
  volume 2879 of {\em Lecture Notes in Computer Science}, pages 918--925.
  Springer, 2003.

\bibitem{Beg:2005qf}
M.~Beg, M.~Miller, A.~Trouv{\'e}, and L.~Younes.
\newblock Computing large deformation metric mappings via geodesic flows of
  diffeomorphisms.
\newblock {\em International Journal of Computer Vision}, 61(2):139--157, 2005.

\bibitem{cao:05}
Y.~Cao, M.~Miller, R.~Winslow, and L.~Younes.
\newblock Large deformation diffeomorphic metric mapping of vector fields.
\newblock {\em IEEE Transactions on Medical Imaging}, 24(9):1216--1230, 2005.

\bibitem{chen:05}
L.~Chen and Q.~Shao.
\newblock Stein's method for normal approximation.
\newblock In {\em An introduction to Stein's method. . Lect. Notes Ser. Inst.
  Math. Sci. Natl. Univ. Singap.}, volume~4. Singapore Univ. Press, Singapore,
  2005.

\bibitem{cou:36}
R.~Courant.
\newblock {\em Differential and Integral Calculus}, volume~2.
\newblock Wiley, 1936.

\bibitem{dup:98}
P.~Dupuis, U.~Grenander, and M.~Miller.
\newblock Variational problems on flows of diffeomorphisms for image matching.
\newblock {\em Quarterly of Applied Mathematics}, 56(3):587--600, 1998.

\bibitem{Grenander:1998:CAE:309082.309089}
U.~Grenander and M.~Miller.
\newblock Computational anatomy: an emerging discipline.
\newblock {\em Q. Appl. Math.}, LVI(4):617--694, December 1998.

\bibitem{hor:91}
R.~Horn and C.~Johnson.
\newblock {\em Topics in matrix analysis}.
\newblock Cambridge University Press, New York, 1991.

\bibitem{mcc:95}
R.~McCann.
\newblock Existence and uniqueness of monotone measure-preserving maps.
\newblock {\em Duke Math. J.}, 80(2):309--323, 1995.

\bibitem{Miller:1999bh}
M~Miller, S~Joshi, and G~Christensen.
\newblock Large deformation fluid diffeomorphisms for landmark and image
  matching.
\newblock {\em A. Toga}, Brain Warping:115--132, 1999.

\bibitem{Miller:2006zr}
M.~Miller, A.~Trouv{\'e}, and L.~Younes.
\newblock Geodesic shooting for computational anatomy.
\newblock {\em Journal of Mathematical Imaging and Vision}, 24(2):209--228,
  2006.

\bibitem{Miller:2001ve}
M.~Miller and L.~Younes.
\newblock Group actions, homeomorphisms, and matching: A general framework.
\newblock {\em International Journal of Computer Vision}, 41:61--84, 2001.

\bibitem{rud:66}
W.~Rudin.
\newblock {\em Real and complex analysis}.
\newblock McGraw-Hill, 1966.

\bibitem{stein:81}
C.~Stein.
\newblock Estimation of the mean of a multivariate normal distribution.
\newblock {\em Ann. Statist.}, 9(6):1135--1151, 1981.

\bibitem{stein:04}
C.~Stein, P.~Diaconis, S.~Holmes, and G.~Reinert.
\newblock Use of exchangeable pairs in the analysis of simulations.
\newblock In {\em Stein's method: expository lectures and applications},
  volume~46 of {\em IMS Lecture Notes Monogr. Ser.}, pages 1--26. Inst. Math.
  Statist., Beachwood, OH, 2004.

\bibitem{Trouve:1998ys}
A~Trouv{\'e}.
\newblock Diffeomorphisms groups and pattern matching in image analysis.
\newblock {\em International Journal of Computer Vision}, 28(3):213--221, 1998.

\bibitem{ty:dq}
A.~Trouv{\'e} and Younes L.
\newblock Local geometry of deformable templates.
\newblock {\em SIAM J. on Mathematical Analysis}, 37(1):17--59, 2005.

\bibitem{vaillant:04}
M.~Vaillant, M.~Miller, L.~Younes, and A.~Trouv{\'e}.
\newblock Statistics on diffeomorphisms via tangent space representations.
\newblock {\em NeuroImage}, 23, Supplement 1(0):S161 -- S169, 2004.

\bibitem{wahba:90}
G.~Wahba.
\newblock {\em Spline Models for Observational Data}.
\newblock SIAM, Philadelphia, PA., 1990.

\bibitem{you:10}
L.~Younes.
\newblock {\em Shapes and diffeomorphisms}.
\newblock Springer, Heidelberg, 2010.

\bibitem{Younes:2008}
L.~Younes, A.~Qiu, R.~Winslow, and M.~Miller.
\newblock Transport of relational structures in groups of diffeomorphisms.
\newblock {\em J. Math. Imaging Vis.}, 32(1):41--56, 2008.

\end{thebibliography}

\end{document}